\documentclass[12pt, draftclsnofoot, onecolumn]{IEEEtran}

\usepackage[font={small}]{caption}

\usepackage[dvips]{color}

\usepackage{graphicx}
\usepackage{amssymb,amsmath}
\usepackage{amsthm}
\usepackage{textcomp}
\usepackage{indentfirst}
\usepackage{url}
\usepackage[nosort]{cite}
\usepackage{color}
\usepackage{marvosym}
\usepackage{amsmath}
\usepackage{cite}
\usepackage{algorithmic}
\usepackage{algorithm}
\usepackage{epsf}
\usepackage{epstopdf}
\usepackage{balance}

\usepackage{comment}
\usepackage{todonotes}
\usepackage{epsf}
\usepackage{epsfig}
\usepackage{times}
\usepackage{epsfig}
\usepackage{graphicx}
\usepackage{bbold}
\usepackage{mathtools}
\usepackage{mathrsfs}
\usepackage{amssymb}
\usepackage{pdfpages}
\usepackage{epstopdf}
\usepackage{dsfont}
\usepackage{lettrine} % \lettrine[findent=1pt]{{{R}}}{}
\usepackage{amsmath,epsfig,amssymb,algorithm,amsthm,cite,url}
\usepackage{subcaption}
\allowdisplaybreaks
\usepackage{csquotes}
\usepackage{verbatim}
\usepackage[english]{babel}
\usepackage{amsmath,amssymb}

\usepackage{verbatim}

\newtheorem{theorem}{\bf Theorem}

\newtheorem{proposition}{\bf Proposition}
\newtheorem{lemma}{\bf Lemma}

\newtheorem{definition}{\bf Definition}

\newlength{\aligntop}
\newlength{\alignbot}
\makeatletter

\renewenvironment{align}{%
\vspace{\aligntop}
\start@align\@ne\st@rredfalse\m@ne
}{%
\math@cr \black@\totwidth@
\egroup
\ifingather@
\restorealignstate@
\egroup
\nonumber
\ifnum0=`{\fi\iffalse}\fi
\else
$$%
\fi
\ignorespacesafterend%
\vspace{\alignbot}\par\noindent
}

\begin{document}

\addtolength{\textfloatsep}{-8pt}

\title{Cellular-Connected UAVs over 5G: Deep Reinforcement Learning for Interference Management\vspace{-0.5cm}}

%\title{Deep Reinforcement Learning for Interference-Aware Path Planning of Cellular-Connected UAVs\vspace{-0.3cm}}

\author{
\IEEEauthorblockN{Ursula Challita\IEEEauthorrefmark{1}, Walid Saad\IEEEauthorrefmark{2}, and Christian Bettstetter\IEEEauthorrefmark{3}}\\
\IEEEauthorblockA{\IEEEauthorrefmark{1}\small
School of Informatics,
The University of Edinburgh, Edinburgh, UK. Email: ursula.challita@ed.ac.uk.}\\
%\IEEEauthorblockA{\IEEEauthorrefmark{2}\small
%School of Informatics, The University of Edinburgh, Edinburgh, United Kingdom\\
\IEEEauthorblockA{\IEEEauthorrefmark{2}\small Wireless@VT, Bradley Department of Electrical and Computer Engineering, Virginia Tech, Blacksburg, VA, USA. Email: walids@vt.edu.}\\
\IEEEauthorblockA{\IEEEauthorrefmark{3}\small Alpen-Adria-Universit\"at Klagenfurt, Institute of Networked and Embedded Systems, Klagenfurt, Austria. Email: christian.bettstetter@aau.at.}\vspace{-1cm}
\thanks{A preliminary version of this work has been accepted for publication at the IEEE International Conference on Communications (ICC) 2018~\cite{ICC_paper}.}
}

\vspace{-0.8cm}
\maketitle\vspace{-0.8cm}
\vspace{-0.3cm}
\begin{abstract}
\vspace{-0.2cm}
\boldmath
%Providing wireless cellular connectivity for unmanned aerial vehicle (UAV) user equipments (UEs) is contingent upon proper management of their resulting interference. To this end, in this paper, an interference-aware path planning scheme for a network of cellular-connected UAVs is proposed.
In this paper, an interference-aware path planning scheme for a network of cellular-connected unmanned aerial vehicles (UAVs) is proposed. In particular, each UAV aims at achieving a tradeoff between maximizing energy efficiency and minimizing both wireless latency and the interference level caused on the ground network along its path. The problem is cast as a dynamic game among UAVs. To solve this game, a deep reinforcement learning algorithm, based on echo state network (ESN) cells, is proposed. The introduced deep ESN architecture is trained to allow each UAV to map each observation of the network state to an action, with the goal of minimizing a sequence of time-dependent utility functions. Each UAV uses ESN to learn its optimal path, transmission power level, and cell association vector at different locations along its path. The proposed algorithm is shown to reach a subgame perfect Nash equilibrium (SPNE) upon convergence. Moreover, an upper and lower bound for the altitude of the UAVs is derived thus reducing the computational complexity of the proposed algorithm. Simulation results show that the proposed scheme achieves better wireless latency per UAV and rate per ground user (UE) while requiring a number of steps that is comparable to a heuristic baseline that considers moving via the shortest distance towards the corresponding destinations. The results also show that the optimal altitude of the UAVs varies based on the ground network density and the UE data rate requirements and plays a vital role in minimizing the interference level on the ground UEs as well as the wireless transmission delay of the UAV.
\end{abstract}
%\renewcommand{\thefootnote}{\fnsymbol{footnote}}
%\footnotetext[1]{...}
\IEEEpeerreviewmaketitle
%\vspace{-0.5cm}
%\vspace{-0.13cm}
%


\section{Introduction}
%\vspace{-0.1cm}
%\ucc{elaborate abstract and intro}

Cellular-connected unmanned aerial vehicles (UAVs) will be an integral component of future wireless networks as evidenced by recent interest from academia, industry, and 3GPP standardizations~\cite{3GPP_standards, Qualcomm_UAV, LTEintheSky, SkyNotLimit, coexistence_ground_aerial, christian}. Unlike current wireless UAV connectivity that relies on short-range communication range (e.g., WiFi, bluetooth, and radio waves), cellular-connected UAVs allow beyond line-of-sight control, low latency, real
time communication, robust security, and ubiquitous coverage. Such \emph{cellular-connected UAV-user equipments (UEs)} will thus enable a myriad of applications ranging from real-time video streaming to surveillance. Nevertheless, the ability of UAV-UEs to establish line-of-sight (LoS) connectivity to cellular base stations (BSs) is both a blessing and a curse. On the one hand, it enables high-speed data access for the UAV-UEs. On the other hand, it can lead to substantial inter-cell mutual interference among the UAVs and to the ground users. As such, a wide-scale deployment of UAV-UEs is only possible if interference management challenges are addressed~\cite{LTEintheSky, SkyNotLimit, coexistence_ground_aerial}.

While some literature has recently studied the use of UAVs as mobile BSs~\cite{U_globecom, ferryMessage, zhang_trajectory_power, path_planning_WCNC, mohammad_UAV, qingqing_UAV, chen2016caching}, the performance analysis of cellular-connected UAV-UEs (\emph{short-handed hereinafter as UAVs}) remains relatively scarce~\cite{LTEintheSky, SkyNotLimit, coexistence_ground_aerial, reshaping_cellular}. For instance, in~\cite{LTEintheSky}, the authors study the impact of UAVs on the uplink performance of a ground LTE network. Meanwhile, the work in~\cite{SkyNotLimit} uses measurements and ray tracing simulations to study the airborne connectivity requirements and propagation characteristics of UAVs. The authors in~\cite{coexistence_ground_aerial} analyze the coverage probability of the downlink of a cellular network that serves both aerial and ground users. In~\cite{reshaping_cellular}, the authors consider a network consisting of both ground and aerial UEs and derive closed-form expressions for the coverage probability of the ground and drone UEs. Nevertheless, this prior art is limited to studying the impact that cellular-connected UAVs have on the ground network. Indeed, the existing literature~\cite{LTEintheSky, SkyNotLimit, coexistence_ground_aerial, reshaping_cellular} does not provide any concrete solution for optimizing the performance of a cellular network that serves both aerial and ground UEs in order to overcome the interference challenge that arises in this context. UAV trajectory optimization is essential in such scenarios. An online path planning that accounts for wireless metrics is vital and would, in essence, assist in addressing the aforementioned interference challenges along with new improvements in the design of the network, such as 3D frequency resue. Such a path planning scheme allows the UAVs to adapt their movement based on the rate requirements of both aerial UAV-UEs and ground UEs, thus improving the overall network performance. The problem of UAV path planning has been studied mainly for non-UAV-UE applications~\cite{ferryMessage, zhang_trajectory_power, path_planning_WCNC, networked_camera} with~\cite{path_cellular_UAVs} being the only work considering a cellular-connected UAV-UE scenario. In~\cite{ferryMessage}, the authors propose a distributed path planning algorithm for multiple UAVs to deliver delay-sensitive information to different ad-hoc nodes. The authors in~\cite{zhang_trajectory_power} optimize a UAV's trajectory in an energy-efficient manner. The authors in~\cite{path_planning_WCNC} propose a mobility model that combines area coverage, network connectivity, and UAV energy constraints for path planning. In~\cite{networked_camera}, the authors propose a fog-networking-based system architecture to coordinate a network of UAVs for video services in sports events. However, despite being interesting, the body of work in~\cite{ferryMessage, zhang_trajectory_power, path_planning_WCNC} and~\cite{networked_camera} is restricted to UAVs as BSs and does not account for UAV-UEs and their associated interference challenges. Hence, the approaches proposed therein cannot readily be used for cellular-connected UAVs. On the other hand, the authors in~\cite{path_cellular_UAVs} propose a path planning scheme for minimizing the time required by a cellular-connected UAV to reach its destination. Nevertheless, this work is limited to one UAV and does not account for the interference that cellular-connected UAVs cause on the ground network during their mission. Moreover, the work in~\cite{path_cellular_UAVs} relies on offline optimization techniques that cannot adapt to the uncertainty and dynamics of a cellular network.

%In~\cite{zhang_2}, the authors study the trajectory design problem for a UAV-enabled multicasting system to minimize the mission completion time while ensuring a high probability value for each ground UE to successfully recover the transmitted file.

%which typically varies over time and for different geographical areas.

The main contribution of this paper is a novel deep reinforcement learning (RL) framework based on echo state network (ESN) cells for optimizing the trajectories of multiple cellular-connected UAVs in an online manner. This framework will allow cellular-connected UAVs to minimize the interference they cause on the ground network as well as their wireless transmission latency. To realize this, we propose a dynamic noncooperative game in which the players are the UAVs and the objective of each UAV is to \emph{autonomously} and \emph{jointly} learn its path, transmit power level, and association vector. For our proposed game, the UAV's cell association vector, trajectory optimization, and transmit power level are closely coupled with each other and their optimal values vary based on the dynamics of the network. Therefore, a major challenge in this game is the need for each UAV to have full knowledge of the ground network topology, ground UEs service requirements, and other UAVs' locations. Consequently, to solve this game, we propose a deep RL ESN-based algorithm, using which the UAVs can predict the dynamics of the network and subsequently determine their optimal paths as well as the allocation of their resources along their paths. Unlike previous studies which are either centralized or rely on the coordination among UAVs, our approach is based on a self-organizing path planning and resource allocation scheme. In essence, two important features of our proposed algorithm are \emph{adaptation} and \emph{generalization}. Indeed, UAVs can take decisions for \emph{unseen} network states, based on the reward they got from previous states. This is mainly due to the use of ESN cells which enable the UAVs to retain their previous memory states. We have shown that the proposed algorithm reaches a subgame perfect Nash equilibrium (SPNE) upon convergence. Moreover, upper and lower bounds on the UAVs' altitudes, that guarantee a maximum interference level on the ground network and a maximum wireless transmission delay for the UAV, have been derived. To our best knowledge, \emph{this is the first work that exploits the framework of deep ESN for interference-aware path planning of cellular-connected UAVs}. Simulation results show that the proposed approach improves the tradeoff between energy efficiency, wireless latency, and the interference level caused on the ground network. Results also show that each UAV's altitude is a function of the ground network density and the UAV's objective function and is an important factor in achieving the UAV's target.

The rest of this paper is organized as follows. Section~\ref{system_model} presents the system model. Section~\ref{game} describes the proposed noncooperative game model. The deep RL ESN-based algorithm is proposed in Section~\ref{algorithm}. In Section~\ref{simulation}, simulation results are analyzed. Finally, conclusions are drawn in Section~\ref{conclusion}.

\vspace{-0.5cm}
\section{System Model}\label{system_model}
%\vspace{-0.1cm}
Consider the uplink (UL) of a wireless cellular network composed of a set $\mathcal{S}$ of $S$ ground BSs, a set $\mathcal{Q}$ of $Q$ ground UEs, and a set $\mathcal{J}$ of $J$ cellular-connected UAVs. The UL is defined as the link from UE $q$ or UAV $j$ to BS $s$. Each BS $s \in \mathcal{S}$ serves a set $\mathcal{K}_s\subseteq\mathcal{Q}$ of $K_s$ UEs and a set $\mathcal{N}_s\subseteq\mathcal{J}$ of $N_s$ cellular-connected UAVs. The total system bandwidth, $B$, is divided into a set $\mathcal{C}$ of $C$ resource blocks (RBs). Each UAV $j\in \mathcal{N}_s$ is allocated a set $\mathcal{C}_{j,s}\subseteq\mathcal{C}$ of $C_{j,s}$ RBs and each UE $q\in \mathcal{K}_s$ is allocated a set $\mathcal{C}_{q,s}\subseteq\mathcal{C}$ of $C_{q,s}$ RBs by its serving BS $s$. At each BS $s$, a particular RB $c \in \mathcal{C}$ is allocated to \emph{at most} one UAV $j\in \mathcal{N}_s$, or UE $q\in \mathcal{K}_s$.

An airborne Internet of Things (IoT) is considered in which the UAVs are equipped with different IoT devices, such as cameras, sensors, and GPS that can be used for various applications such as surveillance, monitoring, delivery and real-time video streaming. The 3D coordinates of each UAV $j \ \in \mathcal{J}$ and each ground user $q \ \in \mathcal{Q}$ are, respectively, $(x_j, y_j, h_j)$ and $(x_q, y_q, 0)$. All UAVs are assumed to fly at a fixed altitude $h_j$ above the ground (as done in~\cite{zhang_trajectory_power, path_cellular_UAVs, relaying, optimization}) while the horizonal coordinates $(x_j, y_j)$ of each UAV $j$ vary in time. Each UAV $j$ needs to move from an initial location $o_j$ to a final destination $d_j$ while transmitting \emph{online} its mission-related data such as sensor recordings, video streams, and location updates. We assume that the initial and final locations of each UAV are pre-determined based on its mission objectives.

%In essence the altitude of each UAV plays a role in .. and therefore,

For ease of exposition, we consider a virtual grid for the mobility of the UAVs. We discretize the space into a set $\mathcal{A}$ of $A$ equally sized unit areas. The UAVs move along the center of the areas $c_a=(x_a, y_a, z_a)$, which yields a finite set of possible paths $\boldsymbol{p}_j$ for each UAV $j$. The path $\boldsymbol{p}_j$ of each UAV $j$ is defined as a sequence of area units $\boldsymbol{p}_j=(a_1, a_2, \cdots, a_l)$ such that $a_1=o_j$ and $a_l=d_j$. The area size of the discretized area units $(a_1, a_2, \cdots, a_A) \in \mathcal{A}$ is chosen to be sufficiently small such that the UAVs' locations can be assumed to be approximately constant within each area even at the maximum UAV's speed, as commonly done in the literature~\cite{relaying}. We assume a constant speed $0 < V_j \leq \widehat{V}_j$ for each UAV where $\widehat{V}_j$ is the maximum speed of UAV $j$. Therefore, the time required by each UAV to travel between any two unit areas is constant.% and thus can be discretized into equal-time slots indexed  .% which we denote as  .

\vspace{-0.1cm}
\subsection{Channel Models}
\vspace{-0.1cm}
We consider the sub-6 GHz band and the free-space path loss model for the UAV-BS data link. The path loss between UAV $j$ at location $a$ and BS $s$, $\xi_{j,s,a}$, is given by~\cite{hourani}:
%\vspace{-0.1cm}
\begin{align}
\xi_{j,s,a} (\mathrm{dB})= 20\ \mathrm{log}_{10} (d_{j,s,a}) + 20\ \mathrm{log}_{10} (\hat{f}) - 147.55,
\end{align}
%\vspace{-0.2cm}
\noindent where $\hat{f}$ is the system center frequency and $d_{j,s,a}$ is the Euclidean distance between UAV $j$ at location $a$ and BS $s$. We consider a Rician distribution for modeling the small-scale fading between UAV $j$ and ground BS $s$ thus accounting for the LoS and multipath scatterers that can be experienced at the BS. In particular, adopting the Rician channel model for the UAV-BS link is validated by the fact that the channel between a given UAV and a ground BS is mainly dominated by a LoS link~\cite{zhang_trajectory_power}. We assume that the Doppler spread due to the mobility of the UAVs is compensated for based on existing techniques such as frequency synchronization using a phase-locked loop~\cite{mengali} as done in~\cite{zhang_trajectory_power} and~\cite{relaying}.

%patent: https://www.google.com/patents/US20170041895

%rician matlab code: http://www.rfwireless-world.com/source-code/MATLAB/Rician-channel-model-matlab-code.html
%MATLAB provides built in function by name 'ricianchan'

For the terrestrial UE-BS links, we consider a Rayleigh fading channel. For a carrier frequency, $\hat{f}$, of 2 GHz, the path loss between UE $q$ and BS $s$ is given by~\cite{pathloss_ground}:
%\vspace{-0.1cm}
\begin{align}
\zeta_{q,s}(\mathrm{dB}) =
15.3+37.6\ \mathrm{log}_{10}(d_{q,s}),
\end{align}
%\vspace{-0.2cm}
\noindent where $d_{q\textrm{,}s}$ is the Euclidean distance between UE $q$ and BS $s$.

The average signal-to-interference-plus-noise ratio (SINR), $\Gamma_{j,s,c,a}$, of the UAV-BS link between UAV $j$ at location $a$ $(a \in \mathcal{A})$ and BS $s$ over RB $c$ will be:
\begin{align}\label{SNIR}
\Gamma_{j,s,c,a}=\frac{P_{j,s,c,a} h_{j,s,c,a}}{I_{j,s,c}+B_c N_0},
\end{align}
\noindent where $P_{j,s,c,a}=\widehat{P}_{j,s,a}/C_{j,s}$ is the transmit power of UAV $j$ at location $a$ to BS $s$ over RB $c$ and $\widehat{P}_{j,s,a}$ is the total transmit power of UAV $j$ to BS $s$ at location $a$. Here, the total transmit power of UAV $j$ is assumed to be distributed uniformly among all of its associated RBs. $h_{j,s,c,a}=g_{j,s,c,a}10^{-\xi_{j,s,a}/10}$ is the channel gain between UAV $j$ and BS $s$ on RB $c$ at location $a$ where $g_{j,s,c,a}$ is the Rician fading parameter. $N_0$ is the noise power spectral density and $B_{c}$ is the bandwidth of an RB $c$. $I_{j,s,c}= \sum_{r=1, r\neq s}^S (\sum_{k=1}^{K_r} P_{k,r,c} h_{k,s,c} + \sum_{n=1}^{N_r} P_{n,r,c,a'} h_{n,s,c,a'})$ is the total interference power on UAV $j$ at BS $s$ when transmitting over RB $c$, where $\sum_{r=1, r\neq s}^S \sum_{k=1}^{K_r} P_{k,r,c} h_{k,s,c}$ and $\sum_{r=1, r\neq s}^S\sum_{n=1}^{N_r} P_{n,r,c,a'} h_{n,s,c,a'}$ correspond, respectively, to the interference from the $K_r$ UEs and the $N_r$ UAVs (at their respective locations $a'$) connected to neighboring BSs $r$ and transmitting using the same RB $c$ as UAV $j$. $h_{k,s,c}=m_{k,s,c}10^{-\zeta_{k,s}/10}$ is the channel gain between UE $k$ and BS $s$ on RB $c$ where $m_{k,s,c}$ is the Rayleigh fading parameter. Therefore, the achievable data rate of UAV $j$ at location $a$ associated with BS $s$ can be defined as $R_{j,s,a}=\sum_{c=1}^{C_{j,s}} B_{c} \mathrm{log}_2(1+\Gamma_{j,s,c,a})$. %UAV $j$.

Given the achievable data rate of UAV $j$ and assuming that each UAV is an M/D/1 queueing system, the corresponding latency over the UAV-BS wireless link is given by~\cite{delay_book}:
\begin{align}\label{delay_eqn}
\tau_{j,s,a}=\frac{\lambda_{j,s}}{2\mu_{j,s,a}\textrm{(}\mu_{j,s,a}-\lambda_{j,s}\textrm{)}}+\frac{1}{\mu_{j,s,a}},
\end{align}
\noindent where $\lambda_{j,s}$ is the average packet arrival rate (packets/s) traversing link $(j,s)$ and originating from UAV $j$. $\mu_{j,s,a}=R_{j,s,a}/\nu$
is the service rate over link $(j,s)$ at location $a$ where $\nu$ is the packet size. On the other hand, the achievable data rate for a ground UE $q$ served by BS $s$ is given by:
\vspace{-0.2cm}
\begin{align}
R_{q,s}=\sum_{c=1}^{C_{q,s}}B_c\mathrm{log}_2\Big(1+\frac{P_{q,s,c}h_{q,s,c}}{I_{q,s,c}+B_cN_0}\Big),
\end{align}
\noindent where $h_{q,s,c}=m_{q,s,c}10^{-\zeta_{q,s}/10}$ is the channel gain between UE $q$ and BS $s$ on RB $c$ and $m_{q,s,c}$ is the Rayleigh fading parameter. $P_{q,s,c}=\widehat{P}_{q,s}/C_{q,s}$ is the transmit power of UE $q$ to its serving BS $s$ on RB $c$ and $\widehat{P}_{q,s}$ is the total transmit power of UE $q$. Here, we also consider equal power allocation among the allocated RBs for the ground UEs. $I_{q,s,c}= \sum_{r=1, r\neq s}^S (\sum_{k=1}^{K_r} P_{k,r,c} h_{k,s,c} + \sum_{n=1}^{N_r} P_{n,r,c,a'} h_{n,s,c,a'})$ is the total interference power experienced by UE $q$ at BS $s$ on RB $c$ where $\sum_{r=1, r\neq s}^S \sum_{k=1}^{K_r} P_{k,r,c} h_{k,s,c}$ and $\sum_{r=1, r\neq s}^S\sum_{n=1}^{N_r} P_{n,r,c,a'} h_{n,s,c,a'}$ correspond, respectively, to the interference from the $K_r$ UEs and the $N_r$ UAVs (at their respective locations $a'$) associated with the neighboring BSs $r$ and transmitting using the same RB $c$ as UE $q$.
%\ucc{add Rayleigh fading also. or check if our model accounts for it or not}

\subsection{Problem Formulation}
Our objective is to find the optimal path for each UAV $j$ based on its mission objectives as well as its interference on the ground network. Thus, we seek to minimize: a) the interference level that each UAV causes on the ground UEs and other UAVs, b) the transmission delay over the wireless link, and c) the time needed to reach the destination. To realize this, we optimize the paths of the UAVs jointly with the cell association vector and power control at each location $a \in \mathcal{A}$ along each UAV's path.
%It should be noted that battery outage, UAV collision, and navigation control issues (e.g., upon encountering static or mobile obstacles) are out of scope of this work.
We consider a directed graph $G_j=(\mathcal{V}, \mathcal{E}_j)$ for each UAV $j$ where $\mathcal{V}$ is the set of vertices corresponding to the centers of the unit areas $a \in \mathcal{A}$ and $\mathcal{E}_j$ is the set of edges formed along the path of UAV $j$. We let $\boldsymbol{\widehat{P}}$ be the transmission power vector with each element $\widehat{P}_{j,s,a}\in[0, \overline{P}_{j}]$ being the transmission power level of UAV $j$ to its serving BS $s$ at location $a$ where $\overline{P}_{j}$ is the maximum transmission power of UAV $j$. $\boldsymbol{\alpha}$ is the path formation vector with each element $\alpha_{j,a,b}\in\{0,1\}$ indicating whether or not a directed link is formed from area $a$ towards area $b$ for UAV $j$, i.e., if UAV $j$ moves from $a$ to $b$ along its path. $\boldsymbol{\beta}$ is the UAV-BS association vector with each element $\beta_{j,s,a}\in\{0,1\}$ denoting whether or not UAV $j$ is associated with BS $s$ at location $a$. Next, we present our optimization problem whose goal is to determine the path of each UAV along with its cell association vector and its transmit power level at each location $a$ along its path $\boldsymbol{p}_j$:
%\begin{align}\label{obj}
%\max_{\mathbf{\boldsymbol{P}_{j,s,a}, \boldsymbol{\alpha_{j,a,b}}}} \sum_{j=1}^J\sum_{a\in \boldsymbol{p}_j}\Big(R_{j,s,a} - \tau_{j,s,a}\Big) + \sum_{q=1}^Q R_{q,s},
%\end{align}
%\ucc{add time}
\vspace{-0.4cm}
\begin{multline}\label{obj}
\hspace{-0.5cm}\min_{\boldsymbol{\widehat{P}}, \boldsymbol{\alpha}, \boldsymbol{\beta}}\vartheta\sum_{j=1}^{J}\sum_{s=1}^S \sum_{c=1}^{C_{j,s}}\sum_{a=1}^A \sum_{r=1, r\neq s}^S \frac{\widehat{P}_{j,s,a} h_{j,r,c,a}}{C_{j,s}}+\varpi \sum_{j=1}^J\sum_{a=1}^A\sum_{b=1, b\neq a}^A\alpha_{j,a,b}
+ \phi\sum_{j=1}^J\sum_{s=1}^S\sum_{a=1}^A\beta_{j,s,a}\tau_{j,s,a},
%\ucc{- \sum_{j=1}^{J}\sum_{a=1}^A\sum_{b=1, b\neq a}^A\sum_{s=1}^S\sum_{c=1}^{C_j}\beta_{j,s,a}\Gamma_{j,s,c,a}}\\
\end{multline}
\vspace{-0.4cm}
\begin{align}\label{cons_1}
\sum_{b=1, b\neq a}^A\alpha_{j,b,a} \leq 1 \;\;\forall j\in \mathcal{J}, a\in \mathcal{A},
\end{align}
\vspace{-0.45cm}
\begin{align}\label{cons_2}
\sum_{a=1, a\neq o_j}^A\alpha_{j,o_j,a} \textrm{=} 1 \;\;\forall j\in \mathcal{J}, \sum_{a=1, a\neq d_j}^A\alpha_{j,a,d_j} \textrm{=} 1 \;\;\forall j\in \mathcal{J},
\end{align}
\vspace{-0.4cm}
\begin{align}\label{cons_3}
\sum_{a\textrm{=}1, a\neq b}^A\alpha_{j,a,b}-\sum_{f\textrm{=}1, f\neq b}^A\alpha_{j,b,f}\textrm{=} 0 \;\forall j\in \mathcal{J},b\in \mathcal{A} \;(b\neq o_j, b\neq d_j),
\end{align}
\vspace{-0.5cm}
\begin{align}\label{cons_4}
\widehat{P}_{j,s,a}\geq\sum_{b=1, b\neq a}^A\alpha_{j,b,a} \;\;\forall j\in \mathcal{J}, s\in \mathcal{S}, a\in \mathcal{A},
\end{align}
\vspace{-0.74cm}
\begin{align}\label{cons_44}
\widehat{P}_{j,s,a}\geq\beta_{j,s,a} \;\;\forall j\in \mathcal{J}, s\in \mathcal{S}, a\in \mathcal{A},
\end{align}
%\vspace{-0.3cm}
%\begin{align}\label{cons_5}
%\widehat{P}_{j,s,a}\leq\overline{P}_j \;\;\;\forall j\in \mathcal{J}, s\in \mathcal{S}, a\in \mathcal{A},
%\end{align}
\vspace{-0.75cm}
\begin{align}\label{cons_6}
\sum_{s=1}^S \beta_{j,s,a} - \sum_{b=1, b\neq a}^A\alpha_{j,b,a}=0\;\;\;\forall j\in \mathcal{J}, a\in A,
\end{align}
%\vspace{-0.3cm}
%\begin{align}\label{cons_6}
%\sum_{a=1}^A\sum_{s=1}^S\beta_{j,s,a}\tau_{j,s,a}\leq\overline{\tau}_j \;\;\;\forall j\in \mathcal{J},
%\end{align}
\vspace{-0.37cm}
\begin{align}\label{cons_7}
\sum_{c=1}^{C_{j,s}}\Gamma_{j,s,c,a}\geq\beta_{j,s,a}\overline{\Gamma}_j
%\sum_{b=1, b\neq a}^A\alpha_{j,b,a}
\;\;\;\forall j\in \mathcal{J}, s\in \mathcal{S}, a\in \mathcal{A},
\end{align}
\vspace{-0.68cm}
%\begin{align}\label{cons_8}
%\Gamma_{q,s}>\overline{\Gamma}_q\;\;\;\forall q, s,
%\end{align}
%\vspace{-0.3cm}
%remove this constraint since the SINR value of UEs depends also on other factors not just on the path of the UAV
\begin{align}\label{cons_8}
0\leq \widehat{P}_{j,s,a}\leq \overline{P}_{j} \;\;\forall j\in \mathcal{J}\textrm{,} s\in \mathcal{S}\textrm{,} \; a\in \mathcal{A},
\end{align}
\vspace{-1.03cm}
\begin{align}\label{cons_9}
\alpha_{j\textrm{,}a\textrm{,}b}\in\{0\textrm{,}1\}\textrm{,}\;  \beta_{j\textrm{,}s\textrm{,}a}\in\{0\textrm{,}1\} \;\;\forall j\in \mathcal{J}\textrm{,}\; s\in \mathcal{S}\textrm{,} \;\;a,b\in \mathcal{A}\textrm{.}
\end{align}
The objective function in (\ref{obj}) captures the total interference level that the UAVs cause on neighboring BSs along their paths, the length of the paths of the UAVs, and their wireless transmission delay. $\vartheta$, $\varpi$ and $\phi$ are multi-objective weights used to control the tradeoff between the three considered metrics. These weights can be adjusted to meet the requirements of each UAV's mission. For instance, the time to reach the destination is critical in search and rescue applications while the latency is important for online video streaming applications. (\ref{cons_1}) guarantees that each area $a$ is visited by UAV $j$ at most once along its path $\boldsymbol{p}_j$. (\ref{cons_2}) guarantees that the trajectory of each UAV $j$ starts at its initial location $o_j$ and ends at its final destination $d_j$. (\ref{cons_3}) guarantees that if UAV $j$ visits area $b$, it should also leave from area $b$ $(b\neq o_j, b\neq d_j)$. (\ref{cons_4}) and (\ref{cons_44}) guarantee that UAV $j$ transmits to BS $s$ at area $a$ with power $\widehat{P}_{j,s,a}>0$ only if UAV $j$ visits area $a$, i.e., $a\in \boldsymbol{p}_j$ and such that $j$ is associated with BS $s$ at location $a$. (\ref{cons_6}) guarantees that each UAV $j$ is associated with one BS $s$ at each location $a$ along its path $\boldsymbol{p}_j$. (\ref{cons_7}) guarantees an upper limit, $\overline{\Gamma}_j$, for the SINR value $\Gamma_{j,s,c,a}$ of the transmission link from UAV $j$ to BS $s$ on RB $c$ at each location $a$, $a\in \boldsymbol{p}_j$. This, in turn, ensures successful decoding of the transmitted packets at the serving BS. The value of $\overline{\Gamma}_j$ is application and mission specific. Note that the SINR check at each location $a$ is valid for our problem since we consider small-sized area units. (\ref{cons_8}) and (\ref{cons_9}) are the feasibility constraints. The formulated optimization problem is a mixed integer non-linear program, which is computationally complex to solve for large networks.

To address this challenge, we adopt a distributed approach in which each UAV decides autonomously on its next path location along with its corresponding transmit power and association vector. In fact, a centralized approach requires control signals to be transmitted to the UAVs at all time. This might incur high round-trip latencies that are not desirable for real-time applications such as online video streaming. Further, a centralized approach requires a central entity to have full knowledge of the current state of the network and the ability to communicate with all UAVs at all time. However, this might not be feasible in case the UAVs belong to different operators or in scenarios in which the environment changes dynamically. Therefore, we next propose a distributed approach for each UAV $j$ to learn its path $\boldsymbol{p}_j$ along with its transmission power level and association vector at each location $a$ along its path in an autonomous and online manner.

%based on the congestion level at each location along its path, the availability of an LTE connection as well as the target regions that need to be covered.
%\newpage
%\vspace{-0.35cm}
\section{Towards a Self-Organizing Network of an Airborne Internet of Things}\label{game}
%\vspace{-0.1cm}
\subsection{Game-Theoretic Formulation}
%\vspace{-0.15cm}
Our objective is to develop a distributed approach that allows each UAV to take actions in an autonomous and online manner. For this purpose, we model the multi-agent path planning problem as a finite dynamic noncooperative game model $\mathcal{G}$ with perfect information~\cite{walid_book}. Formally, we define the game as $\mathcal{G}=(\mathcal{J}, \mathcal{T}, \mathcal{Z}_j, \mathcal{V}_j, \Pi_j, u_j)$ with the set $\mathcal{J}$ of UAVs being the agents. $\mathcal{T}$ is a finite set of stages which correspond to the steps required for all UAVs to reach their sought destinations. $\mathcal{Z}_j$ is the set of actions that can be taken by UAV $j$ at each $t \in \mathcal{T}$, $\mathcal{V}_j$ is the set of all observed network states by UAV $j$ up to stage $T$, $\Pi_j$ is a set of probability distributions defined over all $z_j \in \mathcal{Z}_j$, and $u_j$ is the payoff function of UAV $j$. At each stage $t \in \mathcal{T}$, the UAVs take actions simultaneously. In particular, each UAV $j$ aims at determining its path $\boldsymbol{p}_j$ to its destination along with its optimal transmission power and cell association vector for each location $a \in \mathcal{A}$ along its path $\boldsymbol{p}_j$. Therefore, at each $t$, UAV $j$ chooses an action $z_j(t) \in \mathcal{Z}_j$ composed of the tuple $\boldsymbol{z}_j(t)=(\boldsymbol{a}_j(t), \widehat{P}_{j,s,a}(t), \boldsymbol{\beta}_{j,s,a}(t))$, where $\boldsymbol{a}_j(t)$=\{left, right, forward, backward, no movement\} corresponds to a fixed step size, $\widetilde{a}_j$, in a given direction. $\widehat{P}_{j,s,a}(t)=[\widehat{P}_{1}, \widehat{P}_{2}, \cdots, \widehat{P}_{O}]$ corresponds to $O$ different maximum transmit power levels for each UAV $j$ and $\boldsymbol{\beta}_{j,s,a}(t)$ is the UAV-BS association vector.

%\begin{corollary}
%\emph{The total number of possible actions for each UAV $j$, $\mid\mathcal{Z}_j\mid$, for each UAV can be expressed as $\mid\mathcal{Z}_j\mid=\mid\boldsymbol{a}_j\mid \times \mid\widehat{P}_{j,s,a}\mid \times \mid\boldsymbol{\beta}_{j,s,a}\mid$.}
%\end{corollary}

For each UAV $j$, let $\mathcal{L}_j$ be the set of its $L_j$ nearest BSs. The observed network state by UAV $j$ at stage $t$, $\boldsymbol{v}_j(t) \in \mathcal{V}_j$, is:
%\vspace{-0.45cm}
\begin{align}\label{input}
\boldsymbol{v}_j(t)\textrm{=}\Big[\{\delta_{j\textrm{,}l\textrm{,}a}(t)\textrm{,} \theta_{j\textrm{,}l\textrm{,}a}(t)\}_{l=1}^{L_j}\textrm{,} \theta_{j\textrm{,}d_j\textrm{,}a}(t)\textrm{,} \{x_j(t)\textrm{,} y_j(t)\}_{j \in \mathcal{J}} \Big]\textrm{,}
\end{align}
\noindent where $\delta_{j,l,a}(t)$ is the Euclidean distance from UAV $j$ at location $a$ to BS $l$ at stage $t$, $\theta_{j,l,a}$ is the orientation angle in the xy-plane from UAV $j$ at location $a$ to BS $l$ defined as $\mathrm{tan}^{-1}(\Delta y_{j,l}/\Delta x_{j,l})$~\cite{orientation_angle} where $\Delta y_{j,l}$ and $\Delta x_{j,l}$ correspond to the difference in the $x$ and $y$ coordinates of UAV $j$ and BS $l$, $\theta_{j,d_j,a}$ is the orientation angle in the xy-plane from UAV $j$ at location $a$ to its destination $d_j$ defined as $\mathrm{tan}^{-1}(\Delta y_{j,d_j}/\Delta x_{j,d_j})$, and $\{x_j(t)\textrm{,} y_j(t)\}_{j \in \mathcal{J}}$ are the horizonal coordinates of all UAVs at stage $t$.
For our model, we consider different range intervals for mapping each of the orientation angle and distance values, respectively, into different states.
%For the orientation angle, we consider 4 different intervals (in degrees) given as: $\geq 315 \;\mathrm{or} <45$, [45 135[, [135, 225[, and [225, 315[. The distance value is discretized into 4 different rages (in m): [0 75[, [75 150[, [150, 225[, and [225, $\infty$[. Also, we consider 4 different ranges for the interference level (in dBm) given by: ]-$\infty$, -100], ]-110, -80], ]-80, -50], and ]-50, $\infty$[. The SINR value is discretized into 4 intervals (in dB) defined as: ]-$\infty$, -6], ]-6, -4], ]-4 0], and ]0, $\infty$[.

%\begin{corollary}
%\emph{The total number of possible states for each UAV is given by:}
%\end{corollary}

Moreover, based on the optimization problem defined in (\ref{obj})-(\ref{cons_9}) and by incorporating the Lagrangian penalty method into the utility function definition for the SINR constraint~(\ref{cons_7}), the resulting utility function for UAV $j$ at stage $t$, $u_j(\boldsymbol{v}_j(t), \boldsymbol{z}_j(t), \boldsymbol{z}_{-j}(t))$, will be given by:%, $(\forall j \in \mathcal{J})$:
\vspace{-0.6cm}

%\begin{figure*}[t!]
%\begin{strip}
\begin{align}\label{utility_t}
u_j(\boldsymbol{v}_j(t), \boldsymbol{z}_j(t), \boldsymbol{z}_{-j}(t))\textrm{=}
\begin{cases}
\Phi(\boldsymbol{v}_j(t)\textrm{,} \boldsymbol{z}_j(t)\textrm{,} \boldsymbol{z}_{-j}(t)) \textrm{+} C\textrm{,} \; \mathrm{if} \; \delta_{j,d_j,a}(t)<\delta_{j,d_j,a'}(t-1)\textrm{,}\\
\Phi(\boldsymbol{v}_j(t)\textrm{,} \boldsymbol{z}_j(t)\textrm{,} \boldsymbol{z}_{-j}(t))\textrm{,} \; \mathrm{if} \; \delta_{j,d_j,a}(t)=\delta_{j,d_j,a'}(t-1)\textrm{,}\\
\Phi(\boldsymbol{v}_j(t)\textrm{,} \boldsymbol{z}_j(t)\textrm{,} \boldsymbol{z}_{-j}(t)) \textrm{-} C \textrm{,} \; \mathrm{if} \; \delta_{j,d_j,a}(t)>\delta_{j,d_j,a'}(t-1)\textrm{,}
\end{cases}
\end{align}
%\end{strip}
%\end{figure*}
\vspace{-0.3cm}
\noindent where $\Phi(\boldsymbol{v}_j(t)\textrm{,} \boldsymbol{z}_j(t)\textrm{,} \boldsymbol{z}_{-j}(t))$ is defined as:
\vspace{-0.1cm}
\begin{multline}
\hspace{-0.3cm}\Phi(\boldsymbol{v}_j(t)\textrm{,} \boldsymbol{z}_j(t)\textrm{,} \boldsymbol{z}_{-j}(t))\textrm{=}-\vartheta' \sum_{c=1}^{C_{j,s}(t)} \sum_{r=1, r\neq s}^S \frac{\widehat{P}_{j,s,a}(\boldsymbol{v}_j(t)) h_{j,r,c,a}(t)}{C_{j,s}(t)} - \phi'\tau_{j,s,a}(\boldsymbol{v}_j(t)\textrm{,} \boldsymbol{z}_j(t)\textrm{,} \boldsymbol{z}_{-j}(t)) \\- \varsigma (\mathrm{min}(0, \sum_{c=1}^{C_{j,s}(t)}\Gamma_{j,s,c,a}(\boldsymbol{v}_j(t)\textrm{,} \boldsymbol{z}_j(t)\textrm{,} \boldsymbol{z}_{-j}(t))-\overline{\Gamma}_j))^2,
\end{multline}

\noindent subject to (\ref{cons_1})-(\ref{cons_6}), (\ref{cons_8}) and (\ref{cons_9}). $\varsigma$ is the penalty coefficient for (\ref{cons_7}) and $C$ is a constant parameter. $a'$ and $a$ are the locations of UAV $j$ at $(t-1)$ and $t$ where  $\delta_{j,d_j,a}$ is the distance between UAV $j$ and its destination $d_j$. It is worth noting here that the action space of each UAV $j$ and, thus, the complexity of the proposed game $\mathcal{G}$ increases exponentially when updating the 3D coordinates of the UAVs. Nevertheless, each UAV's altitude must be bounded in order to guarantee an SINR threshold for the UAV and a minimum achievable data rate for the ground UEs. Next, we derive an upper and lower bound for the optimal altitude of any given UAV $j$  based on the proposed utility function in (\ref{utility_t}). In essence, such bounds are valid for all values of the multi-objective weights $\vartheta '$, $\phi '$, and $\varsigma$.

\begin{theorem}\label{theorem_altitude}
\emph{For all values of $\vartheta '$, $\phi '$, and $\varsigma$, a given network state $\boldsymbol{v}_j(t)$, and a particular action $\boldsymbol{z}_j(t)$, the upper and lower bounds for the altitude of UAV $j$ are, respectively, given by:}
\begin{align}
h_j^{\mathrm{max}}(\boldsymbol{v}_j(t)\textrm{,} \boldsymbol{z}_j(t)\textrm{,} \boldsymbol{z}_{-j}(t))= \mathrm{max} (\chi, \hat{h}_j^{\mathrm{max}}(\boldsymbol{v}_j(t)\textrm{,} \boldsymbol{z}_j(t)\textrm{,} \boldsymbol{z}_{-j}(t))),
\end{align}
\vspace{-0.9cm}
\begin{align}
h_j^{\mathrm{min}}(\boldsymbol{v}_j(t)\textrm{,} \boldsymbol{z}_j(t)\textrm{,} \boldsymbol{z}_{-j}(t))= \mathrm{max} (\chi, \hat{h}_j^{\mathrm{min}}(\boldsymbol{v}_j(t)\textrm{,} \boldsymbol{z}_j(t)\textrm{,} \boldsymbol{z}_{-j}(t))),
\end{align}

\noindent \emph{where $\chi$ corresponds to the minimum altitude at which a UAV can fly. $\hat{h}_j^{\mathrm{max}}(\boldsymbol{v}_j(t)\textrm{,} \boldsymbol{z}_j(t)\textrm{,} \boldsymbol{z}_{-j}(t))$ and $\hat{h}_j^{\mathrm{min}}(\boldsymbol{v}_j(t)\textrm{,} \boldsymbol{z}_j(t)\textrm{,} \boldsymbol{z}_{-j}(t))$} \emph{are expressed as:}
\vspace{-0.3cm}
\begin{multline}
\hat{h}_j^{\mathrm{max}}(\boldsymbol{v}_j(t)\textrm{,} \boldsymbol{z}_j(t)\textrm{,} \boldsymbol{z}_{-j}(t))= \\ \sqrt{\frac{\widehat{P}_{j,s,a}(\boldsymbol{v}_j(t))}{C_{j,s}(t) \cdot \overline{\Gamma}_j \cdot \left(\frac{4 \pi \hat{f}}{\hat{c}}\right)^2} \cdot \sum_{c=1}^{C_{j,s}(t)}\frac{g_{j,s,c,a}(t)}{I_{j,s,c}(t)+B_cN_0} - (x_j - x_s)^2 - (y_j - y_s)^2},
\end{multline}
\emph{and}
\begin{align}
\hat{h}_j^{\mathrm{min}}(\boldsymbol{v}_j(t)\textrm{,} \boldsymbol{z}_j(t)\textrm{,} \boldsymbol{z}_{-j}(t))= \max_r \hat{h}_{j,r}^{\mathrm{min}}(\boldsymbol{v}_j(t)\textrm{,} \boldsymbol{z}_j(t)\textrm{,} \boldsymbol{z}_{-j}(t)),
\end{align}
\emph{where
$\hat{h}_{j,r}^{\mathrm{min}}(\boldsymbol{v}_j(t)\textrm{,} \boldsymbol{z}_j(t)\textrm{,} \boldsymbol{z}_{-j}(t))$ is the minimum altitude that UAV $j$ should operate at with respect to a particular neighboring BS $r$ and is expressed as:}
\begin{align}
\hat{h}_{j,r}^{\mathrm{min}}(\boldsymbol{v}_j(t)\textrm{,} \boldsymbol{z}_j(t)\textrm{,} \boldsymbol{z}_{-j}(t))= \sqrt{\frac{\widehat{P}_{j,s,a}(\boldsymbol{v}_j(t)) \cdot \sum_{c=1}^{C_{j,s}(t)} g_{j,r,c,a}(t)}{C_{j,s}(t) \cdot \left(\frac{4 \pi \hat{f}}{\hat{c}}\right)^2 \cdot  \sum_{c=1}^{C_{j,s}(t)} \bar{I}_{j,r,c,a}} - (x_j - x_r)^2 - (y_j - y_r)^2},
\end{align}
\end{theorem}

\begin{proof}
See Appendix A.
\end{proof}

From the above theorem, we can deduce that the optimal altitude of the UAVs is a function of their objective function, location of the ground BSs, network design parameters, and the interference level from other UEs and UAVs in the network. Therefore, at each time step $t$, UAV $j$ would adjust its altitude level based on the values of $h_j^{\mathrm{max}}(\boldsymbol{v}_j(t)\textrm{,} \boldsymbol{z}_j(t)\textrm{,} \boldsymbol{z}_{-j}(t)$ and $h_j^{\mathrm{min}}(\boldsymbol{v}_j(t)\textrm{,} \boldsymbol{z}_j(t)\textrm{,} \boldsymbol{z}_{-j}(t)$ thus adapting to the dynamics of the network. In essence, the derived upper and lower bounds for the optimal altitude of the UAVs allows a reduction of the action space of game $\mathcal{G}$ thus simplifying the process needed for the UAVs to find a solution, i.e., equilibrium, of the game. Next, we analyze the equilibrium point of the proposed game $\mathcal{G}$.

%Moreover, this range for the altitude of the UAVs is valid for all values of $\vartheta '$, $\phi '$, and $\varsigma$ which therefore makes .. more generic for different applications. and thus limits the

%$\boldsymbol{\widehat{z}}_j = (\boldsymbol{z}_j(1), \boldsymbol{z}_j(2), \cdots, \boldsymbol{z}_j(T))$ and $\boldsymbol{\widehat{z}}_{-j}=(\boldsymbol{z}_{-j}(1), \boldsymbol{z}_{-j}(2), \cdots, \boldsymbol{z}_{-j}(T))$ correspond to the sequence of all actions of UAV $j$ and other UAVs $(-j)$ over their corresponding paths $\boldsymbol{p}_j$ and $\boldsymbol{p}_{-j}$, respectively.

%Here, note that we consider a discretized action space for each UAV. In particular, the transmission power for each UAV, $\widehat{P}_{j,s,a}$, belongs to finite set of $l$ transmission power levels, $\widehat{P}_{j,s,a} \in\{\widehat{P}_1, \widehat{P}_2, \cdots, \widehat{P}_l\}$.

\vspace{-0.3cm}
\subsection{Equilibrium Analysis}
\vspace{-0.1cm}
For our game $\mathcal{G}$, we are interested in studying the subgame perfect Nash equilibrium (SPNE) in behavioral strategies. An SPNE is a profile of strategies which induces a Nash equilibrium (NE) on every subgame of the original game. Moreover, a \emph{behavioral strategy} allows each UAV to assign independent probabilities to the set of actions at each network state that is independent across different network states.
Here, note that there always exists at least one SPNE for any finite horizon extensive game with perfect information [Selten's Theorem]~\cite{SPNE_existence}. Let $\boldsymbol{\pi}_j(\boldsymbol{v}_j(t))=(\pi_{j,z_1}(\boldsymbol{v}_j(t)), \pi_{j,z_2}(\boldsymbol{v}_j(t)), \cdots, \pi_{j,\boldsymbol{z}_{\mid Z_j\mid}}(\boldsymbol{v}_j(t))) \in \Pi_j$ be the behavioral strategy of UAV $j$ at state $\boldsymbol{v}_j(t)$ and let $\Delta (\mathcal{Z})$ be the set of all probability distributions over the action space $\mathcal{Z}$. Next, we define the notion of an SPNE.

\begin{definition}\emph{A behavioral strategy $(\boldsymbol{\pi}^*_1(\boldsymbol{v}_j(t))\textrm{,} \cdots\textrm{,} \boldsymbol{\pi}_J^*(\boldsymbol{v}_j(t))) = (\boldsymbol{\pi}_j^*(\boldsymbol{v}_j(t)), \boldsymbol{\pi}^*_{-j}(\boldsymbol{v}_j(t)))$ constitutes a} subgame perfect Nash equilibrium \emph{if, $\forall j \in \mathcal{J}$, $\forall t \in \mathcal{T}$ and $\forall \boldsymbol{\pi}_j(\boldsymbol{v}_j(t)) \in \Delta (\mathcal{Z})$, $\overline{u}_j(\boldsymbol{\pi}^*_j(\boldsymbol{v}_j(t)), \boldsymbol{\pi}^*_{-j}(\boldsymbol{v}_j(t)))\geq \overline{u}_j(\boldsymbol{\pi}_j(\boldsymbol{v}_j(t)), \boldsymbol{\pi}^*_{-j}(\boldsymbol{v}_j(t)))$.}
\end{definition}

Therefore, each state $\boldsymbol{v}_j(t)$ and stage $t$, the goal of each UAV $j$ is to maximize its expected sum of discounted rewards, which is computed as the summation of the immediate reward for a given state along with the expected discounted utility of the next states:
\vspace{-0.4cm}
\begin{multline}\label{expected_utility}
\overline{u}(\boldsymbol{v}_j(t), \boldsymbol{\pi}_j(\boldsymbol{v}_j(t))\textrm{,} \boldsymbol{\pi}_{\textrm{-}j}(\boldsymbol{v}_j(t)))=\mathds{E}_{\boldsymbol{\pi}_j(t)}\left\{\sum_{l=0}^\infty \gamma^{l} u_j(\boldsymbol{v}_j(t+l)\textrm{,} \boldsymbol{z}_j(t+l)\textrm{,} \boldsymbol{z}_{\textrm{-}j}(t+l))| \boldsymbol{v}_{j,0}=\boldsymbol{v}_j\right\}\\
\textrm{=}\sum_{\boldsymbol{z}\in\mathcal{Z}} \sum_{l=0}^\infty \gamma^{l} u_j(\boldsymbol{v}_j(t+l)\textrm{,} \boldsymbol{z}_j(t+l)\textrm{,} \boldsymbol{z}_{\textrm{-}j}(t+l)) \prod_{j=1}^J \pi_{j\textrm{,}z_j}(\boldsymbol{v}_j(t+l))\textrm{,}
\end{multline}
\vspace{-0.2cm}

\noindent where $\gamma^l \in (0, 1)$ is a discount factor for delayed rewards and $\mathds{E}_{\boldsymbol{\pi}_j(\boldsymbol{v}_j(t))}$ denotes an expectation over trajectories of states and actions, in which actions are selected according to $\boldsymbol{\pi}_j(\boldsymbol{v}_j(t))$. Here, $\boldsymbol{u}_j$ is the short-term reward for being in state $\boldsymbol{v}_j$ and $\boldsymbol{\overline{u}}_j$ is the expected long-term total reward from state $\boldsymbol{v}_j$ onwards.

Here, note that the UAV's cell association vector, trajectory optimization, and transmit power level are closely coupled with each other and their corresponding optimal values vary based on the UAVs' objectives. In a multi-UAV network, each UAV must have full knowledge of the future reward functions at each information set and thus for all future network states in order to find the SPNE. This in turn necessitates knowledge of all possible future actions of all UAVs in the network and becomes challenging as the number of UAVs increases. To address this challenge, we rely on deep recurrent neural networks (RNNs)~\cite{RNN_survey}. In essence, RNNs exhibit dynamic temporal behavior and are characterized by their adaptive memory that enables them to store necessary previous state information to predict future actions. On the other hand, deep neural networks are capable of dealing with large datasets. Therefore, next, we develop a novel deep RL based on ESNs, a special kind of RNN, for solving the SPNE of our game $\mathcal{G}$.
%Therefore, upon convergence, the strategy learns to map an observation to a planning computation relevant to a given network state, and thus generates action predictions based on the resulting plan. This in turn leads to policies that generalize better to unseen network domains.

\section{Deep Reinforcement Learning for Online Path Planning and Resource Management}\label{algorithm}
In this section, we first introduce a deep ESN-based architecture that allows the UAVs to store previous states whenever needed while being able to learn future network states. Then, we propose an RL algorithm based on the proposed deep ESN architecture to learn an SPNE for our proposed game.

%an online deep reinforcement learning algorithm based on ESN for enabling UAVs to learn hierarchical representations of time-series featured by multiple timescales dynamics.

%%%JOURNAL
%In particular, reinforcement learning can be defined as ``learning from interaction'', or more specifically as ``a computational approach to understanding and automating goal-directed learning and decision-making''~\cite{sutton}. In case the state and action spaces are small, one can consider a tabular approach such as Q-learning and learn explicitly $Q^*(v,z)$ for all state-action pairs $(v,z) \in \mathcal{V}\times \mathcal{Z}$. However, for our proposed problem the state space is large and thus a tabular approach is not possible. Therefore, we consider a functional approximation approach which aims at learning the parameters $\varphi$ such that $f(v,z; \varphi)\approx Q^*(v,z)$ and the number of parameters, $\varphi$, is much less that the state-action pairs. In what follows, we propose a deep learning ESN-based architecture for finding a SPNE for our proposed game. Then, we propose a reinforcement learning approach based on the proposed deep ESN-based architecture.
%%%

\subsection{Deep ESN Architecture}\label{ESN_architecture}
%%%JOURNAL%%
%ESNs are a new type of artificial RNNs with feedback connections that belong to the family of reservoir computing (RC). RC represents a state-of-the-art approach for extremely efficient RNN modeling, yielding to the possibility of investigating the influence of architectural aspects on the time-scale dynamics differentiation separately from learning. The basic idea is to transform the low dimensional temporal input into a higher dimensional state, and then train the output connection weights to make the system output the target information. Because only the output weights are altered, training is typically quick and computationally efficient compared to training of other RNNs. Moreover, RC has proved to be a useful tool for analyzing the intrinsic properties of stacked architectures in RNNs, allowing at the same time to exploit the extreme efficiency of RC training algorithms in the design of novel deep RNN models.
%%%
ESNs are a new type of RNNs with feedback connections that belong to the family of reservoir computing (RC)~\cite{RNN_survey}. An ESN is composed of an input weight matrix $\boldsymbol{W}_{\mathrm{in}}$, a recurrent matrix $\boldsymbol{W}$, and an output weight matrix $\boldsymbol{W}_{\mathrm{out}}$. Because only the output weights are altered, ESN training is typically quick and computationally efficient compared to training other RNNs. Moreover, multiple non-linear reservoir layers can be stacked on top of each other resulting in a \emph{deep ESN architecture}.
%Deep ESNs are essentially composed of a stack of multiple non-linear reservoir layers.
Deep ESNs exploit the advantages of a hierarchical temporal feature representation at different levels of abstraction while preserving the RC training efficiency. They can learn data representations at different levels of abstraction, hence disentangling the difficulties in modeling complex tasks by representing them in terms of simpler ones hierarchically. Let $N_{j,R}^{(n)}$ be the number of internal units of the reservoir of UAV $j$ at layer $n$, $N_{j,U}$ be the external input dimension of UAV $j$ and $N_{j,L}$ be the number of layers in the stack for UAV $j$. Next, we define the following ESN components:
\begin{itemize}
  \item $\boldsymbol{v}_j(t) \in \mathds{R}^{N_{j,U}}$ the external input of UAV $j$ at stage $t$ which effectively corresponds to the current network state,
  \item $\boldsymbol{x}^{(n)}_j(t) \in \mathds{R}^{N_{j,R}^{(n)}}$ as the state of the reservoir of UAV $j$ at layer $n$ at stage $t$,
  \item $\boldsymbol{W}_{j, \mathrm{in}}^{(n)}$ as the input-to-reservoir matrix of UAV $j$ at layer $n$, where $\boldsymbol{W}_{j, \mathrm{in}}^{(n)} \in \mathds{R}^{N_{j,R}^{(n)} \times N_{j,U}}$ for $n=1$, and $\boldsymbol{W}_{j, \mathrm{in}}^{(n)} \in \mathds{R}^{N_{j,R}^{(n)} \times N_{j,R}^{(n-1)}}$ for $n>1$,
  \item $\boldsymbol{W}_j^{(n)} \in \mathds{R}^{N_{j,R}^{(n)} \times N_{j,R}^{(n)}}$ as the recurrent reservoir weight matrix for UAV $j$ at layer $n$,
  \item $\boldsymbol{W}_{j, \mathrm{out}} \in \mathds{R}^{\mid\mathcal{Z}_j\mid \times (N_{j,U}+\sum_{n}N_{j,R}^{(n)})}$ as the reservoir-to-output matrix of UAV $j$ for layer $n$ only.
\end{itemize}

The objective of the deep ESN architecture is to approximate a function $\boldsymbol{F}_j=(F_j^{1}, F_j^{2}, \cdots, F_j^{N_{j,L}})$ for learning an SPNE for each UAV $j$ at each stage $t$. For each $n=1, 2, \cdots, N_{j,L}$, the function $F_j^{(n)}$ describes the evolution of the state of the reservoir at layer $n$, i.e., $\boldsymbol{x_{j}^{(n)}}(t)=F_j^{(n)}(\boldsymbol{v}_j(t), \boldsymbol{x}_j^{(n)}(t-1))$ for $n=1$ and $\boldsymbol{x_{j}^{(n)}}(t)=F_j^{(n)}(\boldsymbol{x}_j^{(n-1)}(t), \boldsymbol{x}_j^{(n)}(t-1))$ for $n>1$. $\boldsymbol{W}_{j, \mathrm{out}}$ and $\boldsymbol{x}^{(n)}_j(t)$ are initialized to zero while $\boldsymbol{W}_{j, \mathrm{in}}^{(n)}$ and $\boldsymbol{W}_j^{(n)}$ are randomly generated. Note that although the dynamic reservoir is initially generated randomly, it is combined later with the external input, $\boldsymbol{v}_j(t)$, in order to store the network states and with the trained output matrix, $\boldsymbol{W}_{j, \mathrm{out}}$, so that it can approximate the reward function. Moreover, the spectral radius of $\boldsymbol{W}_j^{(n)}$ (i.e., the largest eigenvalue in absolute value), $\rho_j^{(n)}$, must be strictly smaller than 1 to guarantee the stability of the reservoir~\cite{echo_state_property}. In fact, the value of $\rho_j^{(n)}$ is related to the variable memory length of the reservoir that enables the proposed deep ESN framework to store necessary previous state information, with larger values of $\rho_j^{(n)}$ resulting in longer memory length.

%\begin{theorem}
%\ucc{add theorem for effect of layering in terms of short-term memory. Results clearly show that deepESN improves the short-term MC of shallow ESN in all the cases.}
%\ucc{optimal tx power}
%\end{theorem}

We next define the deep ESN components: the input and reward functions. For each deep ESN of UAV $j$, we distinguish between two types of inputs: external input, $\boldsymbol{v}_j(t)$, that is fed to the first layer of the deep ESN and corresponds to the current state of the network and input that is fed to all other layers for $n>1$. For our proposed deep ESN, the input to any layer $n>1$ at stage $t$ corresponds to the state of the previous layer, $\boldsymbol{x}_j^{(n-1)}(t)$. Define $\widetilde{u}_j(\boldsymbol{v}_j(t), \boldsymbol{z}_j(t), \boldsymbol{z}_{-j}(t))= u_j(\boldsymbol{v}_j(t), \boldsymbol{z}_j(t), \boldsymbol{z}_{-j}(t)) \prod_{j=1}^J \pi_{j,z_j}(\boldsymbol{v}_j(t))$ as the expected value of the instantaneous utility function $u_j(\boldsymbol{v}_j(t), \boldsymbol{z}_j(t), \boldsymbol{z}_{-j}(t))$ in (\ref{utility_t}) for UAV $j$ at stage $t$. Therefore,
the reward that UAV $j$ obtains from action $\boldsymbol{z}_j$ at a given network state $\boldsymbol{v}_j(t)$:
\vspace{-0.32cm}
\begin{multline}\label{reward}
r_j(\boldsymbol{v}_j(t), \boldsymbol{z}_j(t), \boldsymbol{z}_{-j}(t))
\textrm{=}
\begin{cases}
\widetilde{u}_j(\boldsymbol{v}_j(t), \boldsymbol{z}_j(t), \boldsymbol{z}_{\textrm{-}j}(t)) \textrm{,} \; \mathrm{if\; UAV} \;j\; \mathrm{reaches} \; d_j\textrm{,}\\
\widetilde{u}_j(\boldsymbol{v}_j(t), \boldsymbol{z}_j(t), \boldsymbol{z}_{\textrm{-}j}(t))\textrm{+}\gamma \mathrm{max}_{\boldsymbol{z}_j \in \mathcal{Z}_j} \boldsymbol{W}_{j\textrm{,} \mathrm{out}}(\boldsymbol{z}_j(t \textrm{+} 1)\textrm{,} t \textrm{+}1) \\ \hspace{0.4 cm} [\boldsymbol{v}'_j(t),
\boldsymbol{x}'^{(1)}_j(t), \boldsymbol{x}'^{(2)}_j(t), \cdots, \boldsymbol{x}'^{(n)}_j(t)]\textrm{,} \; \mathrm{otherwise}\textrm{.}
\end{cases}
\end{multline}\raisetag{3\baselineskip}

\noindent Here, $\boldsymbol{v}'_j(t+1)$ and $\boldsymbol{x}'^{(n)}_j(t)$, correspond, respectively, to the next network state and reservoir state of layer $(n)$, at stage $(t+1)$, upon taking actions $\boldsymbol{z}_j(t)$ and $\boldsymbol{z}_{-j}(t)$ at stage $t$. Fig.~\ref{Deep_ESN} shows the proposed reservoir architecture of the deep ESN consisting of two layers.

\begin{figure}[t!]
  \begin{center}
  %\hspace*{-3.3cm}
  \centering
  \vspace{-0.1cm}
   \includegraphics[width=13cm]{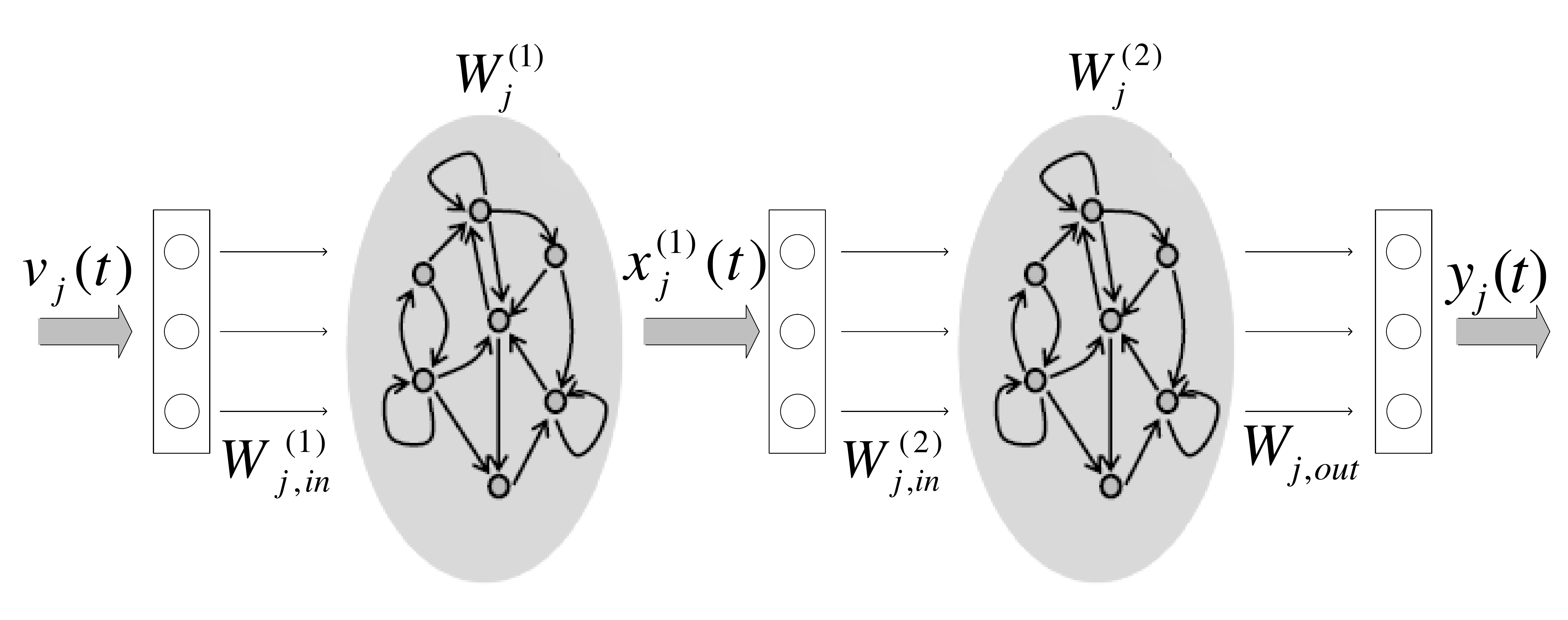} %\hspace*{10cm}
   \vspace{-0.5cm}
   \caption{Proposed Deep ESN architecture.}\label{Deep_ESN}
    \vspace{-0.7cm}
  \end{center}
\end{figure}

%\subsubsection{Output}
%The output of the deep ESN architecture of each UAV $j$ converges to a vector $\overline{\boldsymbol{u}}_j(\boldsymbol{z}_j)$ that represents the expected value of the utility function for each action $z_j$.

%%%\subsubsection{ESN reservoir}
%%%For our proposed deep ESN architecture, we adopt a
%%%
%%%\ucc{\begin{proposition}
%%%check architecture for initializing the reservoir.
%%%\end{proposition}}

\subsection{Update Rule Based on Deep ESN}
We now introduce the deep ESN's update phase that each UAV uses to store and estimate the reward function of each path and resource allocation scheme at a given stage $t$. In particular, we consider leaky integrator reservoir units~\cite{leaky_integrator} for updating the state transition functions $\boldsymbol{x}^{(n)}_j(t)$ at stage $t$. Therefore, the state transition function of the first layer $\boldsymbol{x}^{(1)}_j(t)$ will be:
%\vspace{-0.3cm}
\begin{align}\label{state_1}
\boldsymbol{x}^{(1)}_j(t)= (1-\omega_j^{(1)})\boldsymbol{x}_j^{(1)}(t-1)+\omega_j^{(1)}\mathrm{tanh}(\boldsymbol{W}_{j, \mathrm{in}}^{(1)}\boldsymbol{v}_j(t)+\boldsymbol{W}_j^{(1)}\boldsymbol{x}_j^{(1)}(t-1)),
\end{align}

\noindent where $\omega_j^{(n)} \in [0, 1]$ is the leaking parameter at layer $n$ for UAV $j$ which relates to the speed of the reservoir dynamics in response to the input, with larger values of $\omega_j^{(n)}$ resulting in a faster response of the corresponding $n$-th reservoir to the input. The state transition of UAV $j$, $\boldsymbol{x}^{(n)}_j(t)$, for $n>1$ is given by:
%\vspace{-0.3cm}
\begin{align}\label{state_n}
\boldsymbol{x}^{(n)}_j(t)= (1-\omega_j^{(n)})\boldsymbol{x}_j^{(n)}(t-1)+\omega_j^{(n)}\mathrm{tanh}(\boldsymbol{W}_{j,\mathrm{in}}^{(n)}\boldsymbol{x}_j^{(n-1)}(t)+\boldsymbol{W}_j^{(n)}\boldsymbol{x}_j^{(n)}(t-1)),
\end{align}

The output $y_j(t)$ of the deep ESN at stage $t$ is used to estimate the reward of each UAV $j$ based on the current adopted action $\boldsymbol{z}_j(t)$ and $\boldsymbol{z}_{-j}(t)$ of UAV $j$ and other UAVs $(-j)$, respectively, for the current network state $\boldsymbol{v}_j(t)$ after training $\boldsymbol{W}_{j, \mathrm{out}}$. It can be computed as:
\begin{align}\label{output}
y_j(\boldsymbol{v}_j(t), \boldsymbol{z}_j(t))=\boldsymbol{W}_{j, \mathrm{out}}(\boldsymbol{z}_j(t), t) [\boldsymbol{v}_j(t), \boldsymbol{x}^{(1)}_j(t), \boldsymbol{x}^{(2)}_j(t), \cdots, \boldsymbol{x}^{(n)}_j(t)].
\end{align}

We adopt a temporal difference RL approach for training the output matrix $W_{j, \mathrm{out}}$ of the deep ESN architecture. In particular, we employ a linear gradient descent approach using the reward error signal, given by the following update rule~\cite{RL_ESN}:
\begin{multline}\label{W_out}
\hspace{-0.2cm}\boldsymbol{W}_{j\textrm{,} \mathrm{out}}(\boldsymbol{z}_j(t)\textrm{,} t\textrm{+}1)\textrm{=}\boldsymbol{W}_{j\textrm{,} \mathrm{out}}(\boldsymbol{z}_j(t)\textrm{,} t)\textrm{+}\lambda_j (
r_j(\boldsymbol{v}_j(t)\textrm{,} \boldsymbol{z}_j(t)\textrm{,} \boldsymbol{z}_{\textrm{-}j}(t)) -y_j(\boldsymbol{v}_j(t)\textrm{,} \boldsymbol{z}_j(t)))
[\boldsymbol{v}_j(t)\textrm{,} \\ \boldsymbol{x}^{(1)}_j(t)\textrm{,} \boldsymbol{x}^{(2)}_j(t)\textrm{,} \cdots\textrm{,} \boldsymbol{x}^{(n)}_j(t)]^T\textrm{.}
\end{multline}

Here, note that the objective of each UAV is to minimize the value of the error function $e_j(\boldsymbol{v}_j(t))= \left| r_j(\boldsymbol{v}_j(t)\textrm{,} \boldsymbol{z}_j(t)\textrm{,} \boldsymbol{z}_{\textrm{-}j}(t)) - y_j(\boldsymbol{v}_j(t)\textrm{,} \boldsymbol{z}_j(t))\right|$.

\subsection{Proposed Deep RL Algorithm}
Based on the proposed deep ESN architecture and update rule, we next introduce a multi-agent deep RL framework that the UAVs can use to learn an SPNE in behavioral strategies for the game $\mathcal{G}$. The algorithm is divided into two phases: \emph{training and testing}. In the former, UAVs are trained offline before they become active in the network using the architecture of Subsection~\ref{ESN_architecture}. The testing phase corresponds to the actual execution of the algorithm after which the weights of $\boldsymbol{W}_{j, \mathrm{out}}, \forall j \in \mathcal{J}$ have been optimized and is implemented on each UAV for execution during run time.

\begin{algorithm}[t!] \scriptsize
\caption{Training phase of the proposed deep RL algorithm}
\label{training_algorithm}
\begin{algorithmic}[t!]
\STATE \textbf{Initialization:}\\
$\boldsymbol{\pi}_{j,z_j}(\boldsymbol{v}_j(t))=\frac{1}{\mid \mathcal{A}_j\mid} \forall t\in T, z_j \in \mathcal{Z}_j$, $y_j(\boldsymbol{v}_j(t), \boldsymbol{z}_{j}(t))=0$, $\boldsymbol{W}_{j, \mathrm{in}}^{(n)}$, $\boldsymbol{W}_{j}^{(n)}$, $\boldsymbol{W}_{j, \mathrm{out}}$.
\\
\vspace{0.2cm}
%\STATE \textbf{Deep RL:}
\FOR {The number of training iterations}
\WHILE{At least one UAV $j$ has not reached its destination $d_j$,}
\vspace{0.05cm}
\FOR{all UAVs $j$ (in a parallel fashion)}
\STATE \textbf{Input:} Each UAV $j$ receives an input $\boldsymbol{v}_j(t)$ based on (\ref{input}).
\STATE \textbf{Step 1: Action selection}\\ Each UAV $j$ selects a random action $\boldsymbol{z}_j(t)$ with probability $\epsilon$,\\
Otherwise, UAV $j$ selects $\boldsymbol{z}_j(t)= \mathrm{argmax}_{z_j \in \mathcal{Z}_j} y_{j}\left(\boldsymbol{v}_{j}(t), \boldsymbol{z}_{j}(t)\right)$.
\STATE \textbf{Step 2: Location, cell association and transmit power update}\\
Each UAV $j$ updates its location, cell association and transmission power level based on the selected action $\boldsymbol{z}_j(t)$.\\
\STATE \textbf{Step 3: Reward computation}\\
Each UAV $j$ computes its reward values based on (\ref{reward}).\\
\STATE \textbf{Step 4: Action broadcast}\\ Each UAV $j$ broadcasts its selected action $\boldsymbol{z}_j(t)$ to all other UAVs.\\
\STATE \textbf{Step 5: Deep ESN update}\\
- Each UAV $j$ updates the state transition vector $\boldsymbol{x}_j^{(n)}(t)$ for each layer $(n)$ of the deep ESN architecture based on (\ref{state_1}) and (\ref{state_n}).\\
- Each UAV $j$ computes its output $y_j\left(\boldsymbol{v}_{j}(t), \boldsymbol{z}_{j}(t)\right)$ based on (\ref{output}).\\
- The weights of the output matrix $\boldsymbol{W}_{j,\mathrm{out}}$ of each UAV $j$ are updated based on the linear gradient descent update rule given in (\ref{W_out}).\\
\ENDFOR
\ENDWHILE
\ENDFOR
\end{algorithmic}
\end{algorithm}

During the training phase, each UAV aims at optimizing its output weight matrix $\boldsymbol{W}_{j\textrm{,} \mathrm{out}}$ such that the value of the error function $e_j(\boldsymbol{v}_j(t))$ at each stage $t$ is minimized. In particular, the training phase is composed of multiple iterations, each consisting of multiple rounds, i.e., the number of steps required for all UAVs to reach their corresponding destinations $d_j$. At each round, UAVs face a tradeoff between playing the action associated with the highest expected utility, and trying out all their actions
to improve their estimates of the reward function in (\ref{reward}). This in fact corresponds to the exploration and exploitation tradeoff, in which UAVs need to strike a balance between exploring their environment and exploiting the knowledge accumulated through such exploration~\cite{sutton}. Therefore, we adopt the $\epsilon$-greedy policy in which UAVs choose the action that yields the maximum utility value with a probability of $1- \epsilon + \frac{\epsilon}{\mid \mathcal{Z}_j\mid}$ while exploring randomly other actions with a probability of
$\frac{\epsilon}{\mid\mathcal{A}_j \mid}$. The strategy over the action space will be:
\vspace{-0.1cm}
\begin{align}
\pi_{j,z_j}(\boldsymbol{v}_j(t))=
\begin{cases}
1- \epsilon + \frac{\epsilon}{\mid \mathcal{Z}_j\mid}, \; \mathrm{argmax}_{z_j \in \mathcal{Z}_j} y_{j}\left(\boldsymbol{v}_j(t), \boldsymbol{z}_{j}(t) \right),\\
\frac{\epsilon}{\mid \mathcal{Z}_j\mid}, \; \mathrm{otherwise}.
\end{cases}
\end{align}

Based on the selected action $\boldsymbol{z}_j(t)$, each UAV $j$ updates its location, cell association, and transmission power level and computes its reward function according to (\ref{reward}). To determine the next network state, each UAV $j$ broadcasts its selected action to all other UAVs in the network. Then, each UAV $j$ updates its state transition vector $\boldsymbol{x}_j^{(n)}(t)$ for each layer $(n)$ of the deep ESN architecture according to (\ref{state_1}) and (\ref{state_n}). The output $y_j$ at stage $t$ is then updated based on (\ref{output}). Finally, the weights of the output matrix $\boldsymbol{W}_{j,\mathrm{out}}$ of each UAV $j$ are updated based on the linear gradient descent update rule given in (\ref{W_out}). Note that, a UAV stops taking any actions once it has reached its destination. A summary of the training phase is given in Algorithm~\ref{training_algorithm}.

%\ucc{Interference-Aware Admission Control: upper bound on the number of UAVs in the network to guarantee a threshold value for the interference. for a fixed altitude of the UAVs, }
%\ucc{effect of the step size on the computational complexity. worst case. regret bound.}

\begin{algorithm}[t!] \scriptsize
\caption{Testing phase of the proposed deep RL algorithm}
\label{testing_algorithm}
\begin{algorithmic}[t!]
\vspace{0.2cm}
\WHILE{At least one UAV $j$ has not reached its destination $d_j$,}
\vspace{0.05cm}
\FOR{all UAVs $j$ (in a parallel fashion)}
\STATE \textbf{Input:} Each UAV $j$ receives an input $\boldsymbol{v}_j(t)$ based on (\ref{input}).
\STATE \textbf{Step 1: Action selection}\\ Each UAV $j$ selects an action $\boldsymbol{z}_j(t)= \mathrm{argmax}_{z_j \in \mathcal{Z}_j} y_{j}\left(\boldsymbol{v}_{j}(t), \boldsymbol{z}_{j}(t)\right)$.
\STATE \textbf{Step 2: Location, cell association and transmit power update}\\
Each UAV $j$ updates its location, cell association and transmission power level based on the selected action $\boldsymbol{z}_j(t)$.\\
\STATE \textbf{Step 3: Action broadcast}\\ Each UAV $j$ broadcasts its selected action $\boldsymbol{z}_j(t)$ to all other UAVs.\\
\STATE \textbf{Step 4: State transition vector update}\\
Each UAV $j$ updates the state transition vector $\boldsymbol{x}_j^{(n)}(t)$ for each layer $(n)$ of the deep ESN architecture based on (\ref{state_1}) and (\ref{state_n}).\\
\ENDFOR
\ENDWHILE
\end{algorithmic}
\end{algorithm}

Meanwhile, the testing phase corresponds to the actual execution of the algorithm. In this phase, each UAV chooses its action greedily for each state $\boldsymbol{v}_j(t)$, i.e., $\mathrm{argmax}_{z_j \in \mathcal{Z}_j} y_{j}(\boldsymbol{v}_j(t), \boldsymbol{z}_j(t))$, and updates its location, cell association, and transmission power level accordingly. Each UAV then broadcasts its selected action and updates its state transition vector $\boldsymbol{x}_j^{(n)}(t)$ for each layer $n$ of the deep ESN architecture based on (\ref{state_1}) and (\ref{state_n}). A summary of the testing phase is given in Algorithm \ref{testing_algorithm}.

It is important to note that analytically guaranteeing the convergence of the proposed deep learning algorithm is challenging as it is highly dependent on the hyperparameters used during the training phase. For instance, using too few neurons in the hidden layers results in underfitting which could make it hard for the neural network to detect the signals in a complicated data set. On the other hand, using too many neurons in the hidden layers can either result in overfitting or an increase in the training time that could prevent
the training of the neural network. Overfitting corresponds to the case when the model learns the random fluctuations and noise in the training data set to the extent that it negatively impacts the model's ability to generalize when fed with new data. Therefore, in this work, we limit our analysis of convergence by providing simulation results (see Section~\ref{simulation}) to show that, under a reasonable choice of the hyperparameters, convergence is observed for our proposed game. In such cases, it is important to study the convergence point and the convergence complexity of our proposed algorithm. Next, we characterize the convergence point of our proposed algorithm.

%%Next, we show that, upon convergence, the convergence strategy profile corresponds to an SPNE of game $\mathcal{G}$.

%It is important to note that, upon convergence, the convergence strategy profile corresponds to an SPNE of game $\mathcal{G}$ due to the fact that for any finite game of perfect information, any backward induction solution is an SPNE~\cite{walid_book}.

\begin{proposition}
\emph{If Algorithm~\ref{training_algorithm} converges, then the convergence strategy profile corresponds to a SPNE of game $\mathcal{G}$.}
\end{proposition}

%mention that our game is finite, ... perfect recall game ... backward induction... therefore, based on theorem ... , we converge to SPNE

\begin{proof}
%A dynamic game can be decomposed into a sequence of subgames played at different stages $t$. In essence,
An SPNE is a strategy profile that induces a Nash equilibrium on every subgame. Therefore, at the equilibrium state of each subgame, there is no incentive for any UAV to deviate after observing any history of joint actions. Moreover, given the fact that an ESN framework exhibits adaptive memory that enables it to store necessary previous state information,
UAVs can essentially retain other players' actions at each stage $t$ and thus take actions accordingly. To show that our proposed scheme guarantees convergence to an SPNE, we use the following lemma from~\cite{SPNE_existence}.

\begin{lemma}
For our proposed game $\mathcal{G}$, the payoff functions in (\ref{reward}) are bounded, and the number of players, state space and action space is finite. Therefore, $\mathcal{G}$ is a finite game and hence a SPNE exists. This follows from Selten's theorem which states that every finite extensive form game with perfect recall possesses an SPNE where the
players use behavioral strategies.
\end{lemma}

Here, it is important to note that for finite dynamic games of perfect information, any backward induction solution is a SPNE~\cite{walid_book}. Therefore, given the fact that, for our proposed game $\mathcal{G}$, each UAV aims at maximizing its expected sum of \emph{discounted rewards} at each stage $t$ as given in (\ref{reward}), one can guarantee that the convergence strategy profile corresponds to a SPNE of game $\mathcal{G}$. This completes the proof.
\end{proof}

Moreover, it is important to note that the convergence complexity of the proposed deep RL algorithm for reaching a SPNE is $O(J \times A^2)$. Next, we analyze the computational complexity of the proposed deep RL algorithm for practical scenarios in which the number of UAVs is relatively small.

\begin{theorem}\label{proposition_complexity}
\emph{For practical network scenarios, the computational complexity of the proposed training deep RL algorithm is $O(A^3)$ and reduces to $O(A^2)$ when considering a fixed altitude for the UAVs, where $A$ is the number of discretized unit areas.}
\end{theorem}

\begin{proof}
Consider the case in which the UAVs can move with a fixed step size in a 3D space. For such scenarios, the state vector $\boldsymbol{v}'_j(t)$ is defined as:
\begin{align}\label{input_3D}
\boldsymbol{v}'_j(t)\textrm{=}\Big[\{\delta_{j\textrm{,}l\textrm{,}a}(t)\textrm{,} \theta_{j\textrm{,}l\textrm{,}a}(t)\}_{l=1}^{L_j}\textrm{,} \theta_{j\textrm{,}d_j\textrm{,}a}(t)\textrm{,} \{x_j(t)\textrm{,} y_j(t)\textrm{,} h_j(t)\}_{j \in \mathcal{J}} \Big]\textrm{,}
\end{align}

For each state $\boldsymbol{v}'_j(t)$, the action of UAV $j$ is a function of the location, transmission power level and cell association vector of all other UAVs in the network. Nevertheless, the number of possible locations of other UAVs in the network is much larger than the possible number of transmission power levels and the size of the cell association vector of those UAVs. Therefore, by the law of large numbers, one can consider the number of possible locations of other UAVs only when analyzing the convergence complexity of the proposed training algorithm. Moreover, for practical scenarios, the total number of UAVs in a given area is considered to be relatively small as compared to the number of discretized unit areas i.e., $J \ll A$ (3GPP admission control policy for cellular-connected UAVs~\cite{3GPP_standards}). Therefore, by the law of large numbers and given the fact that the UAVs take actions in a parallel fashion, the computational complexity of our proposed algorithm is $O(A^3)$ when the UAVs update their x, y and z coordinates and reduces to $O(A^2)$ when considering fixed altitudes for the UAVs. This completes the proof.
\end{proof}

From Theorem \ref{proposition_complexity}, we can conclude that the convergence speed of the proposed training algorithm is significantly reduced when considering a fixed altitude for the UAVs. This in essence is due to the reduction of the state space dimension when updating the $x$ and $y$ coordinates only. It is important to note here that there exists a tradeoff between the computational complexity of the proposed training algorithm and the resulting network performance. In essence, updating the 3D coordinates of the UAVs at each step $t$ allows the UAVs to better explore the space thus providing more opportunities for maximizing their corresponding utility functions. Therefore, from both Theorems~\ref{proposition_complexity} and~\ref{theorem_altitude}, the UAVs can update their x and y coordinates only during the learning phase while operating within the upper and lower altitude bounds derived in Theorem~\ref{theorem_altitude}.

\section{Simulation Results and Analysis}\label{simulation}

\begin{table}[t!] \scriptsize
%\footnotesize
%\small
\setlength{\belowcaptionskip}{0pt}
\setlength{\abovedisplayskip}{3pt}
\captionsetup{belowskip=0pt}
\newcommand{\tabincell}[2]{\begin{tabular}{@{}#1@{}}#1.6\end{tabular}}
 \setlength{\abovecaptionskip}{2pt}
 \renewcommand{\captionlabelfont}{\small}
\caption[table]{\scriptsize{\\SYSTEM PARAMETERS}}\label{parameters}
\centering
\tabcolsep=0.06cm %to reduce table width
\scalebox{0.99}{
\begin{tabular}{|c|c|c|c|}
\hline
\textbf{Parameters} & \textbf{Values} & \textbf{Parameters} & \textbf{Values} \\
\hline
UAV max transmit power $(\overline{P}_j)$ & 20 dBm & SINR threshold $(\overline{\Gamma}_j)$ & -3 dB \\
\hline
UE transmit power $(\widehat{P}_q)$ & 20 dBm & Learning rate $(\lambda_j)$ & 0.01\\
\hline
Noise power spectral density $(N_0)$ & -174 dBm/Hz & RB bandwidth $(B_c)$& 180 kHz\\
\hline
Total bandwidth $(B)$ & 20 MHz & \# of interferers $(L)$ & 2\\
\hline
Packet arrival rate $(\lambda_{j,s})$ & (0,1) & Packet size $(\nu)$ & 2000 bits\\
\hline
Carrier frequency $(\hat{f})$ & 2 GHz & Discount factor $(\gamma)$ & 0.7 \\
\hline
\# of hidden layers & 2 & Step size $(\widetilde{a}_j)$ & 40 m \\
\hline
Leaky parameter/layer $(\omega_j^{(n)})$ & 0.99, 0.99 & $\epsilon$ & 0.3\\
\hline
\end{tabular}
}
\vspace{-0.24cm}
\end{table}

For our simulations, we consider an 800 m $\times$ 800 m square area divided into 40 m $\times$ 40 m grid areas, in which we randomly uniformly deploy 15 BSs. All statistical
results are averaged over several independent testing iterations during which the initial locations and destinations of the UAVs and the locations of the BSs and the ground UEs are randomized. The maximum transmit power for each UAV is discretized into 5 equally separated levels. We consider an uncorrelated Rician fading channel with parameter $\widehat{K}=1.59$~\cite{rician_fading}. The external input of the deep ESN architecture, $\boldsymbol{v}_j(t)$, is a function of the number of UAVs and thus the number of hidden nodes per layer, $N_{j,R}^{(n)}$, varies with the number of UAVs. For instance, $N_{j,R}^{(n)}= 12$ and $6$ for $n=1$ and $2$, respectively, for a network size of 1 and 2 UAVs, and 20 and 10 for a network size of 3, 4, and 5 UAVs. Table~\ref{parameters} summarizes the main simulation parameters.
%We consider Monte Carlo simulations in which the initial and destinations of the UAVs, the location of the BSs as well as the ground UEs are randomized for each iteration.

\begin{figure}[!t]
%\vspace{-0.4cm}
\begin{subfigure}{1.0\textwidth}
  \centering
  \includegraphics[width=10cm]{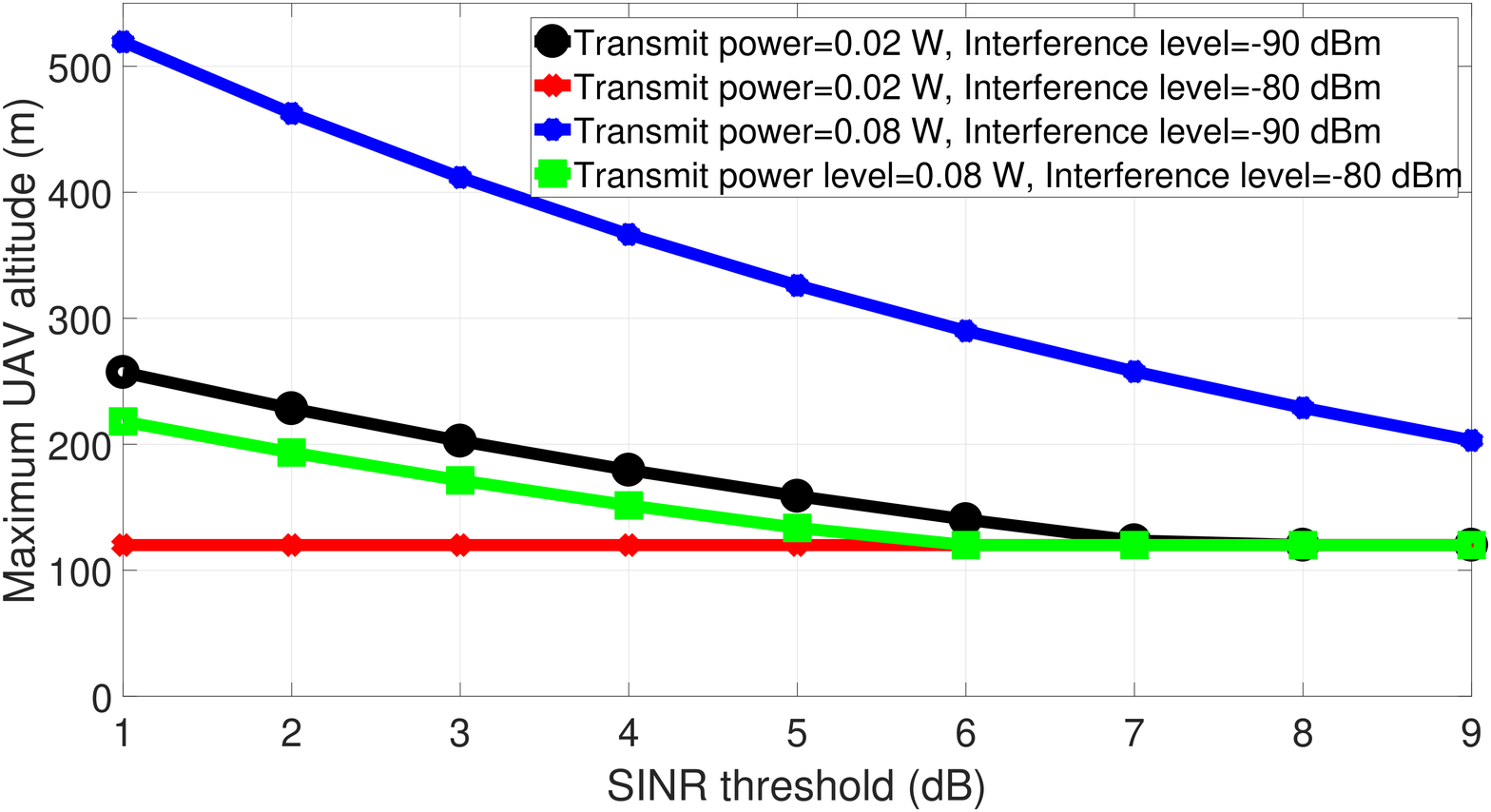}
  \caption{}
  \label{maximum_altitude}
\end{subfigure}\\
%\vspace{0.4cm}
\begin{subfigure}{1.0\textwidth}
  \centering
  \includegraphics[width=10cm]{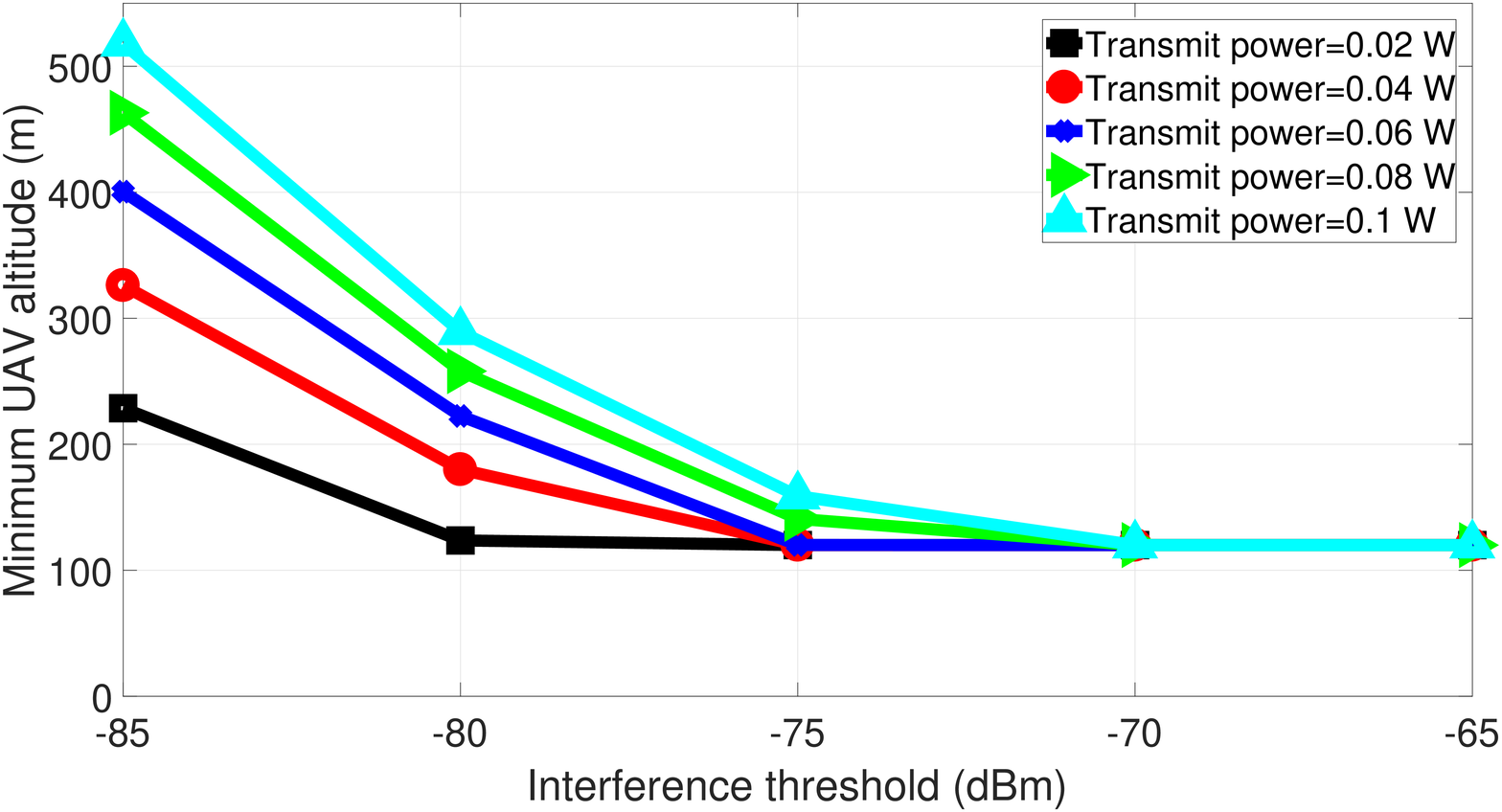}
  \caption{}
  \label{minimum_altitude}
\end{subfigure}
\vspace{-0.3cm}
\caption{The (a) upper bound for the optimal altitude of the UAVs as a function of the SINR threshold value $(\bar{\Gamma})$ and for different transmit power levels and ground network density and (b) lower bound for the optimal altitude of the UAVs as a function of the interference threshold value $(\sum_{c=1}^{C_{j,s}(t)} \bar{I}_{j,r,c,a})$ and for different transmit power levels.}\label{altitude_results}
\vspace{-0.2cm}
\end{figure}

Fig.~\ref{maximum_altitude} shows the upper bound for the optimal altitude of UAV $j$ as a function of the SINR threshold value, $\bar{\Gamma}$, and for different transmit power levels, based on Theorem~\ref{theorem_altitude}. On the other hand, Fig.~\ref{minimum_altitude} shows the lower bound for the optimal altitude of UAV $j$ as a function of the SINR threshold value, $\bar{\Gamma}$, and for different transmit power levels and ground network density, based on Theorem~\ref{theorem_altitude}. From Figs.~\ref{maximum_altitude} and~\ref{minimum_altitude}, we can deduce that the optimal altitude range of a given UAV is a function of network design parameters, ground network data requirements, the density of the ground network, and its action $\boldsymbol{v}_j(t)$. For instance, the upper bound on the UAV's optimal altitude decreases as $\bar{\Gamma}$ increases while its lower bound decreases as $\sum_{c=1}^{C_{j,s}(t)} \bar{I}_{j,r,c,a}$ increases. Moreover, the maximum altitude of the UAV decreases as the ground network gets denser while the its lower bound increases as the ground network data requirements increase. Thus, in such scenarios, a UAV should operate at higher altitudes. A UAV should also operate at higher altitudes when its transmit power level increases due to the increase in the lower and upper bounds of its optimal altitude.

\vspace{-0.1cm}
\begin{figure}[t!]
  \begin{center}
  %\hspace*{-3.3cm}
  \centering
  %\vspace{cm}
   \includegraphics[width=11cm]{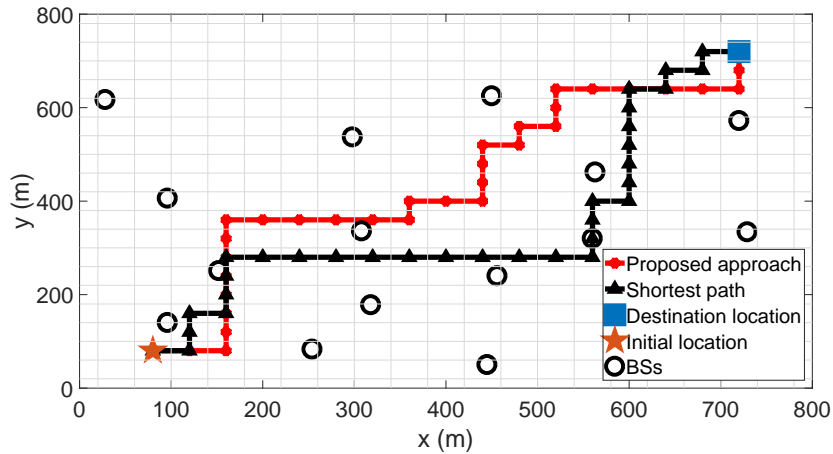} %\hspace*{10cm}
   \vspace{-0.4cm}
   \caption{Path of a UAV for our approach and shortest path scheme.}\label{snapshot}
    \vspace{-0.5cm}
  \end{center}
\end{figure}

\begin{table}[t!]\footnotesize
%\small
\setlength{\belowcaptionskip}{0pt}
\setlength{\abovedisplayskip}{3pt}
\captionsetup{belowskip=0pt}
\newcommand{\tabincell}[2]{\begin{tabular}{@{}#1@{}}#1.6\end{tabular}}
 \setlength{\abovecaptionskip}{2pt}
 \renewcommand{\captionlabelfont}{\small}
\caption[table]{\scriptsize{\\Performance assessment for one UAV}}\label{snapshot_table}
\centering
\tabcolsep=0.1cm %to reduce table width
\scalebox{0.99}{
\begin{tabular}{|c|c|c|c|}
\hline
 & \# of steps & delay (ms) & average rate per UE (Mbps) \\
\hline
Proposed approach & 32 & 6.5 & 0.95\\
\hline
Shortest path & 32 & 12.2 & 0.76\\
\hline
\end{tabular}
}
\vspace{-0.24cm}
\end{table}

Fig.~\ref{snapshot} shows a snapshot of the path of a single UAV resulting from our approach and from a shortest path scheme. Unlike our proposed scheme which accounts for other wireless metrics during path planning, the objective of the UAVs in the shortest path scheme is to reach their destinations with the minimum number of steps. Table~\ref{snapshot_table} presents the performance results for the paths shown in Fig.~\ref{snapshot}. From Fig.~\ref{snapshot}, we can see that, for our proposed approach, the UAV selects a path away from the densely deployed area while maintaining proximity to its serving BS in a way that would minimize the steps required to reach its destination. This path will minimize the interference level that the UAV causes on the ground UEs and its wireless latency (Table~\ref{snapshot_table}). From Table~\ref{snapshot_table}, we can see that our proposed approach achieves 25\% increase in the average rate per ground UE and 47\% decrease in the wireless latency as compared to the shortest path, while requiring the same number of steps that the UAV needs to reach the destination.

\begin{figure}[t!]
  \begin{center}
  %\hspace*{-3.3cm}
  \centering
  %\vspace{-0.5cm}
   \includegraphics[width=11cm, scale=1.9]{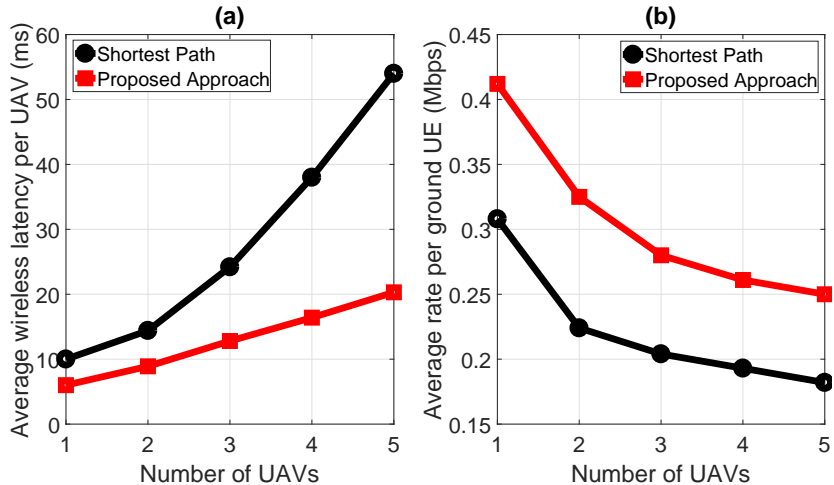} %\hspace*{10cm}
   \vspace{-0.4cm}
   \caption{Performance assessment of the proposed approach in terms of average (a) wireless latency per UAV and (b) rate per ground UE as compared to the shortest path approach, for different number of UAVs.}\label{scalability}
    \vspace{-0.5cm}
    %\vspace{-1.2cm}
  \end{center}
\end{figure}

\begin{table}[t!]\footnotesize
%\small
\setlength{\belowcaptionskip}{0pt}
\setlength{\abovedisplayskip}{3pt}
\captionsetup{belowskip=0pt}
\newcommand{\tabincell}[2]{\begin{tabular}{@{}#1@{}}#1.6\end{tabular}}
 \setlength{\abovecaptionskip}{2pt}
 \renewcommand{\captionlabelfont}{\small}
 \captionsetup{justification=centering}
\caption[table]{\scriptsize{\\The required number of steps for all UAVs to reach their corresponding destinations based on our proposed approach and that of the shortest path scheme for different number of UAVs}}\label{steps_table}
\centering
\tabcolsep=0.1cm %to reduce table width
\scalebox{0.99}{
\begin{tabular}{|c|c|c|c|c|c|}
\hline
\# of steps & 1 UAV & 2 UAVs & 3 UAVs & 4 UAVs & 5 UAVs\\
\hline
Proposed approach & 4 & 4 & 6 & 7 & 8\\
\hline
Shortest path & 4 & 4 & 6 & 6 & 7\\
\hline
\end{tabular}
}
\vspace{-0.24cm}
\end{table}

Fig.~\ref{scalability} compares the average values of the (a) wireless latency per UAV and (b) rate per ground UE resulting from our proposed approach and the baseline shortest path scheme. Moreover, Table~\ref{steps_table} compares the number of steps required by all UAVs to reach their corresponding destinations for the scenarios presented in Fig.~\ref{scalability}. From Fig.~\ref{scalability} and Table~\ref{steps_table}, we can see that, compared to the shortest path scheme, our approach achieves a lower wireless latency per UAV and a higher rate per ground UE for different numbers of UAVs while requiring a number of steps that is comparable to the baseline. In fact, our scheme provides a better tradeoff between energy efficiency, wireless latency, and ground UE data rate compared to the shortest path scheme. For instance, for 5 UAVs, our scheme achieves a 37\% increase in the average achievable rate per ground UE, 62\% decrease in the average wireless latency per UAV, and 14\% increase in energy efficiency. Indeed, one can adjust the multi-objective weights of our utility function based on several parameters such as the rate requirements of the ground network, the power limitation of the UAVs, and the maximum tolerable wireless latency of the UAVs. Moreover, Fig.~\ref{scalability} shows that, as the number of UAVs increases, the average delay per UAV increases and the average rate per ground UE decreases, for all schemes. This is due to the increase in the interference level on the ground UEs and other UAVs as a result of the LoS link between the UAVs and the BSs.

\begin{figure}[t!]
  \begin{center}
  %\hspace*{-3.3cm}
  \centering
  %\vspace{-0.5cm}
   \includegraphics[width=11cm, scale=1.9]{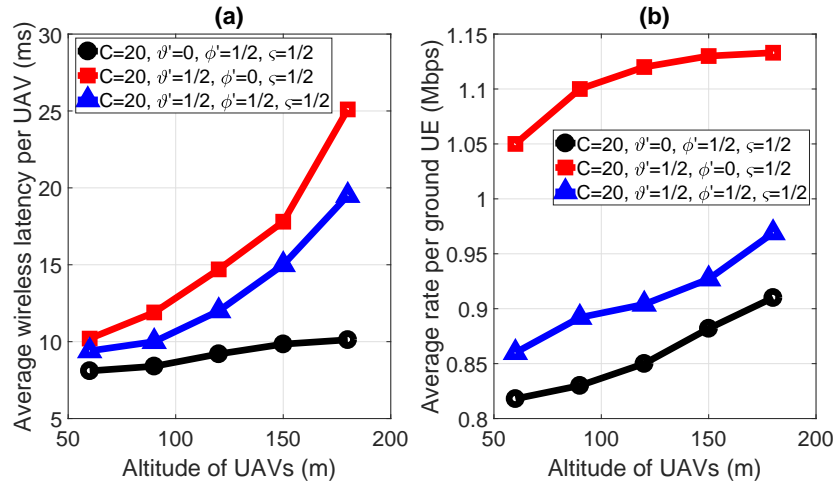} %\hspace*{10cm}
   \vspace{-0.4cm}
   \caption{Performance assessment of the proposed approach in terms of average (a) wireless latency per UAV and (b) rate per ground UE for different utility functions and for different altitudes of the UAVs.}\label{altitude}
    \vspace{-0.9cm}
    %\vspace{-1.2cm}
  \end{center}
\end{figure}

Fig.~\ref{altitude} studies the effect of the UAVs' altitude on the average values of the (a) wireless latency per UAV and (b) rate per ground UE for different utility functions. From Fig.~\ref{altitude}, we can see that, as the altitude of the UAVs increases, the average wireless latency per UAV increases for all studied utility functions. This is mainly due to the increase in the distance of the UAVs from their corresponding serving BSs which accentuates the path loss effect. Moreover, higher UAV altitudes result in a higher average data rate per ground UE for all studied utility functions mainly due to the decrease in the interference level that is caused from the UAVs on neighboring BSs. Here, there exists a tradeoff between minimizing the average wireless delay per UAV and maximizing the average data rate per ground UE. Therefore, alongside the multiobjective weights, the altitude of the UAVs can be varied such that the ground UE rate requirements is met while minimizing the wireless latency for each UAV based on its mission objective.

\begin{figure}[t!]
  \begin{center}
  %\hspace*{-3.3cm}
  \centering
  %\vspace{-0.5cm}
   \includegraphics[width=11cm, scale=1.9]{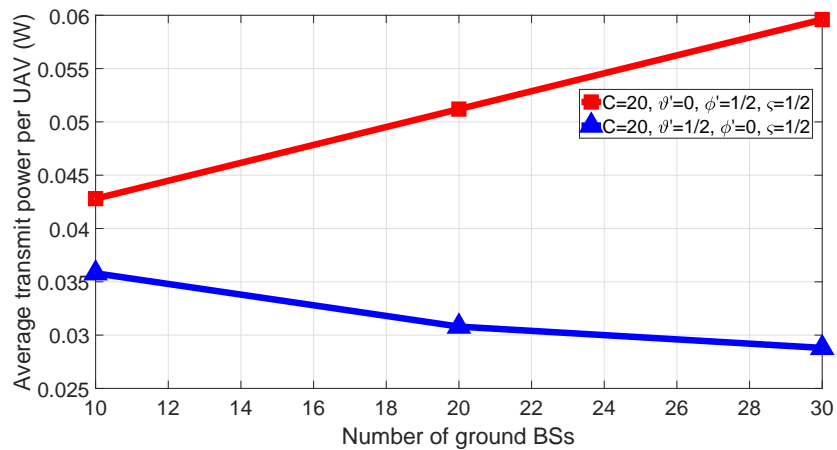} %\hspace*{10cm}
   \vspace{-0.4cm}
   \caption{Effect of the ground network densification on the average transmit power level of the UAVs along their paths.}\label{power_densification}
    \vspace{-0.9cm}
    %\vspace{-1.2cm}
  \end{center}
\end{figure}

Fig.~\ref{power_densification} shows the average transmit power level per UAV along its path as a function of the number of BSs considering two utility functions, one for minimizing the average wireless latency for each UAV and the other for minimizing the interference level on the ground UEs. From Fig.~\ref{power_densification}, we can see that network densification has an impact on the transmission power level of the UAVs. For instance, when minimizing the wireless latency of each UAV along its path, the average transmit power level per UAV increases from 0.04 W to 0.06 W as the number of ground BSs increases from 10 to 30, respectively. In essence, the increase in the transmit power level is the result of the increase in the interference level from the ground UEs as the ground network becomes denser. As a result, the UAVs will transmit using a larger transmission power level so as to minimize their wireless transmission delay. On the other hand, the average transmit power level per UAV decreases from 0.036 W to 0.029 W in the case of minimizing the interference level caused on neighboring BSs. This is due to the fact that as the number of BSs increases, the interference level caused by each UAV on the ground network increases thus requiring each UAV to decrease its transmit power level. Note that, when minimizing the wireless latency, the average transmit power per UAV is always larger than the case of minimizing the interference level, irrespective of the number of ground BSs. Therefore, the transmit power level of the UAVs is a function of their mission objective and the number of ground BSs.

\begin{figure}[t!]
  \begin{center}
  %\hspace*{-3.3cm}
  \centering
  %\vspace{-0.5cm}
   \includegraphics[width=11cm, scale=1.9]{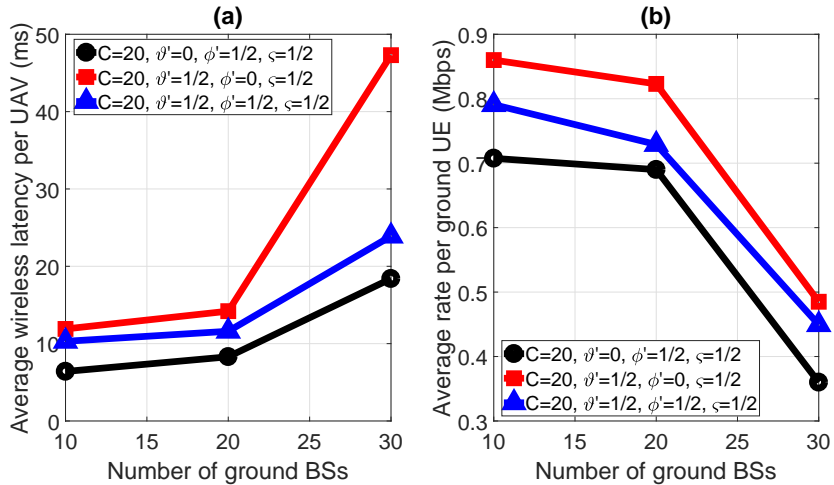} %\hspace*{10cm}
   \vspace{-0.4cm}
   \caption{Effect of the ground network densification on the average (a) wireless latency per UAV and (b) rate per ground UE for different utility functions and for a fixed altitude of 120m.}\label{densification}
    \vspace{-0.7cm}
    %\vspace{-1.2cm}
  \end{center}
\end{figure}

Fig.~\ref{densification} presents the (a) wireless latency per UAV and (b) rate per ground UE for different utilities as a function of the number of BSs and for a fixed altitude of 120 m. From this figure, we can see that, as the ground network becomes more dense, the average wireless latency per UAV increases and the average rate per ground UE decreases for all considered cases. For instance, when the objective is to minimize the interference level along with energy efficiency, the average wireless latency per UAV increases from 13 ms to 47 ms and the average rate per ground UE decreases from 0.86 Mbps to 0.48 Mbps as the number of BSs increases from 10 to 30. This is due to the fact that a denser network results in higher interference on the UAVs as well as other UEs in the network.

\begin{figure}[t!]
  \begin{center}
  %\hspace*{-3.3cm}
  \centering
  %\vspace{-0.5cm}
   \includegraphics[width=11cm, scale=1.9]{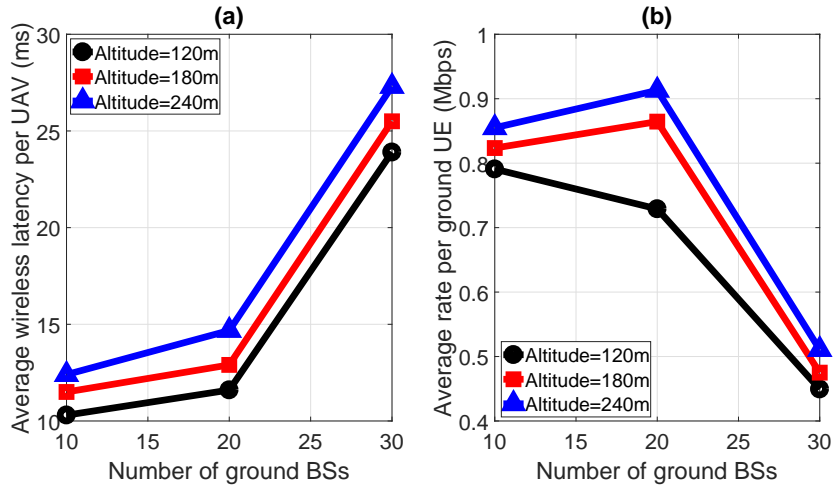} %\hspace*{10cm}
   \vspace{-0.4cm}
   \caption{Effect of the ground network densification on the average (a) wireless latency per UAV and (b) rate per ground UE for different utility functions and for various altitudes of the UAVs.}\label{densification_altitude}
    \vspace{-0.7cm}
    %\vspace{-1.2cm}
  \end{center}
\end{figure}

Fig.~\ref{densification_altitude} investigates the (a) wireless latency per UAV and (b) rate per ground UE for different values of the UAVs altitude and as a function of the number of BSs. From this figure, we can see that as the UAV altitude increases and/or the ground network becomes denser, the average wireless latency per UAV increases. For instance, the delay increases by 27\% as the altitude of the UAVs increases from 120 to 240 m for a network consisting of 20 BSs and increases by 120\% as the number of BSs increases from 10 to 30 for a fixed altitude of 180 m. This essentially follows from Theorem~\ref{theorem_altitude} and the results in Fig.~\ref{maximum_altitude} which shows that the maximum altitude of the UAV decreases as the ground network gets denser and thus the UAVs should operate at a lower altitude when the number of BSs increases from 10 to 30. Moreover, the average rate per ground UE decreases as the ground network becomes denser due to the increase in the interference level and increases as the altitude of the UAVs increases. Therefore, the resulting network performance depends highly on both the UAVs altitude and the number of BSs in the network. For instance, in case of a dense ground network, the UAVs need to fly at a lower altitude for applications in which the wireless transmission latency is more critical and at a higher altitude in scenarios in which a minimum achievable data rate for the ground UEs is required.

\begin{figure}[t!]
  \begin{center}
  %\hspace*{-3.3cm}
  \centering
  %\vspace{-0.5cm}
   \includegraphics[width=11cm, scale=1.9]{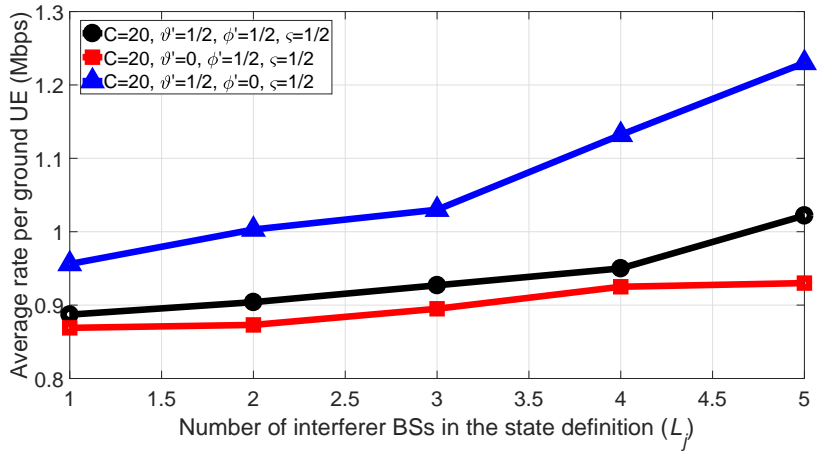} %\hspace*{10cm}
   \vspace{-0.4cm}
   \caption{The average rate per ground UE as a function of the number of interferer BSs in the state definition $(\emph{L}_j)$.}\label{interferers}
    \vspace{-0.9cm}
    %\vspace{-1.2cm}
  \end{center}
\end{figure}

Fig.~\ref{interferers} shows the effect of varying the number of nearest BSs ($\emph{L}_j$) in the observed network state of UAV $j$, $\boldsymbol{v}_j(t)$, on the average data rate per ground UE for different utility functions. From Fig.~\ref{interferers}, we can see an improvement in the average rate per ground UE as the number of nearest BSs in the state definition increases. For instance, in scenarios in which the UAVs aim at minimizing the interference level they cause on the ground network along their paths, the average rate per ground UE increases by 28\% as the number of BSs in the state definition increases from 1 to 5. This gain results from the fact that as $\emph{L}_j$ increases, the UAVs get a better sense of their surrounding environment and thus can better select their next location such that the interference level they cause on the ground network is minimized. It is important to note here, that as $\emph{L}_j$ increases, the size of the external input ($\boldsymbol{v}_j$) increases thus requiring a larger number of neurons in each layer. This in turn increases the number of required iterations for convergence. Therefore, a tradeoff exists between improving the performance of the ground UEs and the running complexity of the proposed algorithm.

%Here, note that in scenarios in which the UAVs aim at minimizing their wireless transmission delay, we notice a slight improvement in the data rate (6\%) as $\emph{L}_j$ increases. In fact, the number of $\emph{L}_j$ BSs does not have much impact on the wireless latency of the UAVs since the UAVs would just care about the nearest few BSs in case of handover. Nevertheless, a large number of $\emph{L}_j$ would not have much effect on the performance since a UAV would typically associate with BSs that are closer.

\begin{figure}[t!]
  \begin{center}
  %\hspace*{-3.3cm}
  \centering
  %\vspace{-0.5cm}
   \includegraphics[width=11cm, scale=1.9]{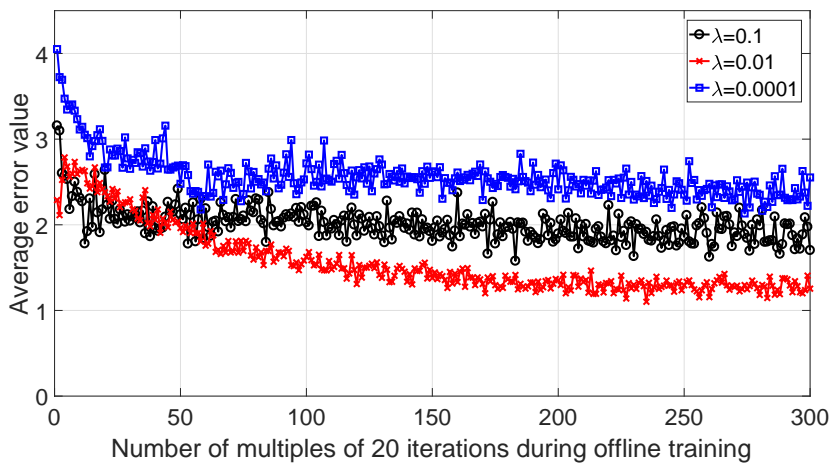} %\hspace*{10cm}
   \vspace{-0.31cm}
   \caption{Effect of the learning rate on the convergence of offline training.}\label{learning_rate}
    \vspace{-0.9cm}
  \end{center}
\end{figure}

Fig.~\ref{learning_rate} shows the average of the error function $e_j(\boldsymbol{v}_j(t))$ resulting from the offline training phase as a function of a multiple of 20 iterations while considering different values for the learning rate, $\lambda$. The learning rate determines the step size the algorithm takes to reach the optimal solution and, thus, it impacts the convergence rate of our proposed framework. From Fig.~\ref{learning_rate}, we can see that small values of the learning rate, i.e., $\lambda =0.0001$, result in a slow speed of convergence. On the other hand, for large values of the learning rate, such as $\lambda=0.1$, the error function decays fast for the first few iterations but then remains constant. Here, $\lambda=0.1$ does not lead to convergence during the testing phase, but $\lambda =0.0001$ and $\lambda=0.01$ result in convergence, though requiring a different number of training iterations. In fact, a large learning rate can cause the algorithm to diverge from the optimal solution. This is because large initial learning rates will decay the loss function faster and thus make the model get stuck at a particular region of the optimization space instead of better exploring it. Clearly, our framework achieves better performance for $\lambda=0.01$, as compared to smaller and larger values of the learning rate. We also note that the error function does not reach the value of zero during the training phase. This is due to the fact that, for our approach, we adopt the early stopping technique to avoid overfitting which occurs when the training error decreases at the expense of an increase in the value of the test error~\cite{RNN_survey}.

%plot the error function and mention that we use early stopping in order to avoid overfitting. to know when to stop, we use a validation set. that's y the error does not reach zero eventually. Here, note that for the case of 0.1, it does not converge even after 10000 iterations (check also what other cases do not converge).bcz there's a tradeoff between minimizing the training error and that of testing. ne need to overcome overfitting so error during training may not necessarily become v low.

\section{Conclusion}\label{conclusion}
\vspace{-0.1cm}
In this paper, we have proposed a novel interference-aware path planning scheme that allows cellular-connected UAVs to minimize the interference they cause on a ground network as well as their wireless transmission latency while transmitting online mission-related data. We have formulated the problem as a noncooperative game in which the UAVs are the players. To solve the game, we have proposed a deep RL algorithm based ESN cells which is guaranteed to reach an SPNE, if it converges. The proposed algorithm enables each UAV to decide on its next location, transmission power level, and cell association vector in an autonomous manner thus adapting to the changes in the network. Simulation results have shown that the proposed approach achieves better wireless latency per UAV and rate per ground UE while requiring a number of steps that is comparable to the shortest path scheme. The results have also shown that a UAV's altitude plays a vital role in minimizing the interference level on the ground UEs as well as the wireless transmission delay of the UAV. In particular, we have shown that the altitude of the UAV is a function of the ground network density, the UAV's objective and the actions of other UAVs in the network.

\section*{Appendix}
\subsection{Proof of Theorem \ref{theorem_altitude}}
For a given network state $\boldsymbol{v}_j(t)$ and a particular action $\boldsymbol{z}_j(t)$, the upper bound for the altitude of UAV $j$ can be derived when UAV $j$ aims at minimizing its delay function only, i.e., $\vartheta '=0$. For such scenarios, UAV $j$ should guarantee an upper limit, $\overline{\Gamma}_j$, for the SINR value $\Gamma_{j,s,c,a}$ of the transmission link from UAV $j$ to BS $s$ on RB $c$ at location $a$ as given in constraint (\ref{cons_7}). Therefore, $\hat{h}_j^{\mathrm{max}}(\boldsymbol{v}_j(t)\textrm{,} \boldsymbol{z}_j(t)\textrm{,} \boldsymbol{z}_{-j}(t))$ corresponds to the altitude at which UAV $j$ achieves $\overline{\Gamma}_j$ and beyond which (\ref{cons_7}) is violated. The derivation of the expression of $\hat{h}_j^{\mathrm{max}}(\boldsymbol{v}_j(t)\textrm{,} \boldsymbol{z}_j(t)\textrm{,} \boldsymbol{z}_{-j}(t))$ is:
\begin{align}
\sum_{c=1}^{C_{j,s}(t)}\Gamma_{j,s,c,a} = \overline{\Gamma}_j,
\end{align}
\begin{align}
\sum_{c=1}^{C_{j,s}(t)} \frac{\frac{\widehat{P}_{j,s,a}(\boldsymbol{v}_j(t))}{C_{j,s}(t)}\cdot g_{j,s,c,a}(t)}{\left(\frac{4 \pi \hat{f} d_{j,s,a}^{\mathrm{max}}}{\hat{c}}\right)^2 \cdot (I_{j,s,c}(t)+B_cN_0)}= \overline{\Gamma}_j,
\end{align}
\begin{align}
\frac{\widehat{P}_{j,s,a}(\boldsymbol{v}_j(t))}{C_{j,s}(t)} \cdot \frac{1}{\left(\frac{4 \pi \hat{f} d_{j,s,a}^{\mathrm{max}}}{\hat{c}}\right)^2} \cdot \sum_{c=1}^{C_{j,s}(t)} \frac{g_{j,s,c,a}{}(t)}{I_{j,s,c}(t)+B_cN_0} = \overline{\Gamma}_j,
\end{align}
\begin{align}
(d_{j,s,a}^{\mathrm{max}})^2=\frac{\widehat{P}_{j,s,a}(\boldsymbol{v}_j(t))}{C_{j,s}(t)} \cdot \frac{1}{\overline{\Gamma}_j \cdot \left(\frac{4 \pi \hat{f}}{\hat{c}}\right)^2} \cdot \sum_{c=1}^{C_{j,s}(t)}\frac{g_{j,s,c,a}(t)}{I_{j,s,c}(t)+B_cN_0},
\end{align}
\noindent where $d_{j,s,a}$ is the Euclidean distance between UAV $j$ and its serving BS $s$ at location $a$. Assume that the altitude of BS $s$ is negligible, i.e., $z_s=0$, $\hat{h}_j^{\mathrm{max}}(\boldsymbol{v}_j(t)\textrm{,} \boldsymbol{z}_j(t)\textrm{,} \boldsymbol{z}_{-j}(t))$ can be expressed as:
\begin{multline}
\hat{h}_j^{\mathrm{max}}(\boldsymbol{v}_j(t)\textrm{,} \boldsymbol{z}_j(t)\textrm{,} \boldsymbol{z}_{-j}(t))= \\ \sqrt{\frac{\widehat{P}_{j,s,a}(\boldsymbol{v}_j(t))}{C_{j,s}(t) \cdot \overline{\Gamma}_j \cdot \left(\frac{4 \pi \hat{f}}{\hat{c}}\right)^2} \cdot \sum_{c=1}^{C_{j,s}(t)}\frac{g_{j,s,c,a}(t)}{I_{j,s,c}(t)+B_cN_0} - (x_j - x_s)^2 - (y_j - y_s)^2},
\end{multline}

\noindent where $x_s$ and $y_s$ correspond to the x and y coordinates of the serving BS $s$ and $\hat{c}$ is the speed of light.

On the other hand, for a given network state $\boldsymbol{v}_j(t)$ and a particular action $\boldsymbol{z}_j(t)$, the lower bound for the altitude of UAV $j$ can be derived when the objective function of UAV $j$ is to minimize the interference level it causes on the ground network only, i.e., $\phi '=0$ and $\varsigma=0$. For such scenarios, the interference level that UAV $j$ causes on neighboring BS $r$ at location $a$ should not exceed a predefined value given by $\sum_{c=1}^{C_{j,s}(t)}\bar{I}_{j,r,c,a}$\footnote{$\sum_{c=1}^{C_{j,s}(t)}\bar{I}_{j,r,c,a}$ is a network design parameter that is a function of the ground network density, number of UAVs in the network and the data rate requirements of the ground UEs. The value of $\bar{I}_{j,r,c,a}$ is in fact part of the admission control policy which limits the number of UAVs in the network and their corresponding interference level on the ground network~\cite{3GPP_standards}.}. Therefore, $\hat{h}_j^{\mathrm{min}}(\boldsymbol{v}_j(t)\textrm{,} \boldsymbol{z}_j(t)\textrm{,} \boldsymbol{z}_{-j}(t))$ corresponds to the altitude at which UAV $j$ achieves $\sum_{c=1}^{C_{j,s}(t)}\bar{I}_{j,r,c,a}$ and below which the level of interference it causes on BS $r$ exceeds the value of $\sum_{c=1}^{C_{j,s}(t)}\bar{I}_{j,r,c,a}$. The derivation of the expression of $\hat{h}_j^{\mathrm{min}}(\boldsymbol{v}_j(t)\textrm{,} \boldsymbol{z}_j(t)\textrm{,} \boldsymbol{z}_{-j}(t))$ is given by:
\begin{align}
\sum_{c=1}^{C_{j,s}(t)}\sum_{r=1, r\neq s}^{S} \frac{\widehat{P}_{j,s,a}(\boldsymbol{v}_j(t)) h_{j,r,c,a}(t)}{C_{j,s}(t)}= \sum_{c=1}^{C_{j,s}(t)} \sum_{r=1, r\neq s}^{S}\bar{I}_{j,r,c,a},
\end{align}
\begin{align}\label{all_interferers}
\sum_{c=1}^{C_{j,s}(t)}\sum_{r=1, r\neq s}^{S} \frac{\widehat{P}_{j,s,a}(\boldsymbol{v}_j(t)) \cdot g_{j,r,c,a}(t)}{C_{j,s}(t) \cdot \left(\frac{4 \pi \hat{f} d_{j,r,a}^{\mathrm{min}}}{\hat{c}}\right)^2 }= \sum_{c=1}^{C_{j,s}(t)} \sum_{r=1, r\neq s}^{S}\bar{I}_{j,r,c,a},
\end{align}
To find $\hat{h}_j^{\mathrm{min}}(\boldsymbol{v}_j(t)\textrm{,} \boldsymbol{z}_j(t)\textrm{,} \boldsymbol{z}_{-j}(t))$, we need to solve (\ref{all_interferers}) for each neighboring BS $r$ separately. Therefore, for a particular neighboring BS $r$, (\ref{all_interferers}) can be written as:
\begin{align}
\sum_{c=1}^{C_{j,s}(t)} \frac{\widehat{P}_{j,s,a}(\boldsymbol{v}_j(t)) \cdot g_{j,r,c,a}(t)}{C_{j,s}(t) \cdot \left(\frac{4 \pi \hat{f} d_{j,r,a}^{\mathrm{min}}}{\hat{c}}\right)^2}= \sum_{c=1}^{C_{j,s}(t)} \bar{I}_{j,r,c,a},
\end{align}
\begin{align}
\frac{\widehat{P}_{j,s,a}(\boldsymbol{v}_j(t)) \cdot \sum_{c=1}^{C_{j,s}(t)} g_{j,r,c,a}(t)}{C_{j,s}(t) \cdot \left(\frac{4 \pi \hat{f} d_{j,r,a}^{\mathrm{min}}}{\hat{c}}\right)^2} = \sum_{c=1}^{C_{j,s}(t)} \bar{I}_{j,r,c,a},
\end{align}
\begin{align}
(d_{j,r,a}^{\mathrm{min}})^2=\frac{\widehat{P}_{j,s,a}(\boldsymbol{v}_j(t)) \cdot \sum_{c=1}^{C_{j,s}(t)} g_{j,r,c,a}(t)}{C_{j,s}(t) \cdot \left(\frac{4 \pi \hat{f}}{\hat{c}}\right)^2 \cdot  \sum_{c=1}^{C_{j,s}(t)} \bar{I}_{j,r,c,a}},
\end{align}
\noindent where $d_{j,r,a}$ is the Euclidean distance between UAV $j$ and its neighboring BS $r$ at location $a$. Assume that the altitude of BS $r$ is negligible, i.e., $z_r=0$, we have:
\begin{align}
\hat{h}_{j,r}^{\mathrm{min}}(\boldsymbol{v}_j(t)\textrm{,} \boldsymbol{z}_j(t)\textrm{,} \boldsymbol{z}_{-j}(t))= \sqrt{\frac{\widehat{P}_{j,s,a}(\boldsymbol{v}_j(t)) \cdot \sum_{c=1}^{C_{j,s}(t)} g_{j,r,c,a}(t)}{C_{j,s}(t) \cdot \left(\frac{4 \pi \hat{f}}{\hat{c}}\right)^2 \cdot  \sum_{c=1}^{C_{j,s}(t)} \bar{I}_{j,r,c,a}} - (x_j - x_r)^2 - (y_j - y_r)^2},
\end{align}
Therefore, $\hat{h}_{j}^{\mathrm{min}}(\boldsymbol{v}_j(t)\textrm{,} \boldsymbol{z}_j(t)\textrm{,} \boldsymbol{z}_{-j}(t))$ corresponds to the maximum value of $\hat{h}_{j,r}^{\mathrm{min}}(\boldsymbol{v}_j(t)\textrm{,} \boldsymbol{z}_j(t)\textrm{,} \boldsymbol{z}_{-j}(t))$ among all neighboring BSs $r$ and is expressed as:
\begin{align}
\hat{h}_j^{\mathrm{min}}(\boldsymbol{v}_j(t)\textrm{,} \boldsymbol{z}_j(t)\textrm{,} \boldsymbol{z}_{-j}(t))= \max_r \hat{h}_{j,r}^{\mathrm{min}}(\boldsymbol{v}_j(t)\textrm{,} \boldsymbol{z}_j(t)\textrm{,} \boldsymbol{z}_{-j}(t)),
\end{align}

\noindent where $x_r$ and $y_r$ correspond to the x and y coordinates of other neighboring BSs $r$. This completes the proof.

%\section{notes}
%check this link:
%-http://www.asctec.de/en/safely-integrating-uavs-flying-beyond-line-of-sight/ \\
%- https://bettstetter.com/talks/thessaloniki-2017/ \\
%
%
%\ucc{For write-up and notations, check smart grid, deva, and the printed paper (multi-robot negotiation: approximating the set of subgame perfect equilibria in general-sum stochastic game) and the 4 deep ESN papers.}

\def\baselinestretch{0.92}
\bibliographystyle{IEEEtran}
\bibliography{references}

\begin{thebibliography}{10}
\providecommand{\url}[1]{#1}
\csname url@samestyle\endcsname
\providecommand{\newblock}{\relax}
\providecommand{\bibinfo}[2]{#2}
\providecommand{\BIBentrySTDinterwordspacing}{\spaceskip=0pt\relax}
\providecommand{\BIBentryALTinterwordstretchfactor}{4}
\providecommand{\BIBentryALTinterwordspacing}{\spaceskip=\fontdimen2\font plus
\BIBentryALTinterwordstretchfactor\fontdimen3\font minus
  \fontdimen4\font\relax}
\providecommand{\BIBforeignlanguage}[2]{{%
\expandafter\ifx\csname l@#1\endcsname\relax
\typeout{** WARNING: IEEEtran.bst: No hyphenation pattern has been}%
\typeout{** loaded for the language `#1'. Using the pattern for}%
\typeout{** the default language instead.}%
\else
\language=\csname l@#1\endcsname
\fi
#2}}
\providecommand{\BIBdecl}{\relax}
\BIBdecl

\bibitem{ICC_paper}
U.~Challita, W.~Saad, and C.~Bettstetter, ``Deep reinforcement learning for
  interference-aware path planning of cellular-connected {UAVs},'' in
  \emph{Proc. of International Conference on Communications (ICC)}.\hskip 1em
  plus 0.5em minus 0.4em\relax Kansas City, MO, USA, May 2018.

\bibitem{3GPP_standards}
\BIBentryALTinterwordspacing
{3GPP}, ``Enhanced {LTE} support for aerial vehicles,'' March 2017. [Online].
  Available:
  \url{https://portal.3gpp.org/desktopmodules/Specifications/SpecificationDetails.aspx?specificationId=3231}
\BIBentrySTDinterwordspacing

\bibitem{Qualcomm_UAV}
\BIBentryALTinterwordspacing
Qualcomm, ``Paving the path to {5G}: Optimizing commercial {LTE} networks for
  drone communication,'' [Online], Sept. 2016. [Online]. Available:
  \url{https://www.qualcomm.com/news/onq/2016/09/06/paving-path-5g-optimizing-commercial-lte-networks-drone-communication}
\BIBentrySTDinterwordspacing

\bibitem{LTEintheSky}
B.~V. der Bergh, A.~Chiumento, and S.~Pollin, ``{LTE} in the sky: Trading off
  propagation benefits with interference costs for aerial nodes,'' \emph{IEEE
  Communications Magazine}, vol.~54, no.~5, pp. 44--50, May 2016.

\bibitem{SkyNotLimit}
X.~Lin, V.~Yajnanarayana, S.~Muruganathan, S.~Gao, and H.~Asplund, ``The sky is
  not the limit: {LTE} for unmanned aerial vehicles,'' \emph{arXiv:1707.07534},
  July 2017.

\bibitem{coexistence_ground_aerial}
M.~Azari, F.~Rosas, A.~Chiumento, and S.~Pollin, ``Coexistence of terrestrial
  and aerial users in cellular networks,'' in \emph{Proc. of IEEE Global
  Communications Conference (GLOBECOM) workshops}.\hskip 1em plus 0.5em minus
  0.4em\relax Singapore, Dec. 2017.

\bibitem{christian}
T.~Andre, K.~Hummel, A.~Schoellig, E.~Yanmaz, M.~Asadpour, C.~Bettstetter,
  P.~Grippa, H.~Hellwagner, S.~Sand, and S.~Zhang, ``Application-driven design
  of aerial communication networks,'' \emph{IEEE Communications Magazine},
  vol.~52, no.~5, pp. 129--137, May 2014.

\bibitem{U_globecom}
U.~Challita and W.~Saad, ``Network formation in the sky: Unmanned aerial
  vehicles for multi-hop wireless backhauling,'' in \emph{Proc. of IEEE Global
  Communications Conference (GLOBECOM)}.\hskip 1em plus 0.5em minus 0.4em\relax
  Singapore, Dec. 2017.

\bibitem{ferryMessage}
J.~Yoon, Y.~Jin, N.~Batsoyol, and H.~Lee, ``Adaptive path planning of {UAVs}
  for delivering delay-sensitive information to ad-hoc nodes,'' in \emph{Proc.
  of IEEE Wireless Communications and Networking Conference (WCNC)}.\hskip 1em
  plus 0.5em minus 0.4em\relax San Francisco, CA, USA, Mar. 2017.

\bibitem{zhang_trajectory_power}
Y.~Zeng and R.~Zhang, ``Energy-efficient {UAV} communication with trajectory
  optimization,'' \emph{IEEE Transactions on Wireless Communications}, vol.~16,
  no.~6, pp. 3747--3760, June 2017.

\bibitem{path_planning_WCNC}
M.~Messous, S.~Senouci, and H.~Sedjelmaci, ``Network connectivity and area
  coverage for {UAV} fleet mobility model with energy constraint,'' in
  \emph{Proc. of IEEE Wireless Communications and Networking Conference
  (WCNC)}.\hskip 1em plus 0.5em minus 0.4em\relax Doha, Qatar, Apr. 2016.

\bibitem{mohammad_UAV}
M.~Mozaffari, W.~Saad, M.~Bennis, and M.~Debbah, ``Unmanned aerial vehicle with
  underlaid device-to-device communications: Performance and tradeoffs,''
  \emph{IEEE Transactions on Wireless Communications}, vol.~15, no.~6, pp.
  3949--3963, June 2016.

\bibitem{qingqing_UAV}
Q.~Wu, J.~Xu, and R.~Zhang, ``Capacity characterization of {UAV}-enabled
  two-user broadcast channel,'' \emph{arXiv:1801.00443}, Jan. 2018.

\bibitem{chen2016caching}
M.~Chen, M.~Mozaffari, W.~Saad, C.~Yin, M.~Debbah, and C.~S. Hong, ``Caching in
  the sky: {P}roactive deployment of cache-enabled unmanned aerial vehicles for
  optimized quality-of-experience,'' \emph{IEEE Journal on Selected Areas on
  Communications (JSAC), Special Issue on Human-In-The-Loop Mobile Networks},
  vol.~35, no.~5, pp. 1046--1061, May 2017.

\bibitem{reshaping_cellular}
M.~Azari, F.~Rosas, and S.~Pollin, ``Reshaping cellular networks for the sky:
  The major factors and feasibility,'' \emph{arXiv:1710.11404}, Oct. 2017.

\bibitem{networked_camera}
X.~Wang, A.~Chowdhery, and M.~Chiang, ``Networked drone cameras for sports
  streaming,'' in \emph{Proc. of International Conference on Distributed
  Computing Systems (ICDCS)}.\hskip 1em plus 0.5em minus 0.4em\relax Atlanta,
  Georgia, USA, June 2017.

\bibitem{path_cellular_UAVs}
S.~Zhang, Y.~Zeng, and R.~Zhang, ``Cellular-enabled {UAV} communication:
  Trajectory optimization under connectivity constraint,''
  \emph{arXiv:1710.11619}, Oct. 2017.

\bibitem{relaying}
Y.~Zeng, R.~Zhang, and T.~Lim, ``Throughput maximization for {UAV}-enabled
  mobile relaying systems,'' \emph{IEEE Transactions on Communications},
  vol.~64, no.~12, pp. 4983--4996, Dec. 2016.

\bibitem{optimization}
M.~Bekhti, M.~Abdennebi, N.~Achir, and K.~Boussetta, ``Path planning of
  unmanned aerial vehicles with terrestrial wireless network tracking,'' in
  \emph{Proc. of Wireless Days}.\hskip 1em plus 0.5em minus 0.4em\relax
  Toulouse, France, May 2016.

\bibitem{hourani}
A.~Hourani, S.~Kandeepan, and A.~Jamalipour, ``Modeling air-to-ground path loss
  for low altitude platforms in urban environments,'' in \emph{Proc. of IEEE
  Global Communications Conference (GLOBECOM)}.\hskip 1em plus 0.5em minus
  0.4em\relax Austin, TX, USA, Dec. 2014.

\bibitem{mengali}
U.~Mengali and A.~D'Andrea, \emph{Synchronization Techniques for Digital
  Receivers}, {Plenum Press}, Ed., New York, 1997.

\bibitem{pathloss_ground}
{3GPP TR 25.942 v2.1.3}, ``3rd generation partnership project; technical
  specification group {(TSG) RAN WG4; RF} system scenarios,'' Tech. Rep., 2000.

\bibitem{delay_book}
D.~Bertsekas and R.~Gallager, \emph{Data Networks}.\hskip 1em plus 0.5em minus
  0.4em\relax Prentice Hall, Mar. 1992.

\bibitem{walid_book}
Z.~Han, D.~Niyato, W.~Saad, T.~Ba\c{s}ar, and A.~Hjorungnes, \emph{Game Theory
  in Wireless and Communication Networks: Theory, Models, and
  Applications}.\hskip 1em plus 0.5em minus 0.4em\relax Cambridge University
  Press, 2012.

\bibitem{orientation_angle}
W.~Kwon, I.~Suh, S.~Lee, and Y.~Cho, ``Fast reinforcement learning using
  stochastic shortest paths for a mobile robot,'' in \emph{Proc. of {IEEE/RSJ}
  International Conference on Intelligent Robots and Systems}.\hskip 1em plus
  0.5em minus 0.4em\relax San Diego, CA, USA, Nov. 2007.

\bibitem{SPNE_existence}
M.~Osborne, \emph{An Introduction to Game Theory}.\hskip 1em plus 0.5em minus
  0.4em\relax Oxford University Press, 2004.

\bibitem{RNN_survey}
M.~Chen, U.~Challita, W.~Saad, C.~Yin, and M.~Debbah, ``Machine learning for
  wireless networks with artificial intelligence: A tutorial on neural
  networks,'' \emph{arXiv:1710.02913}, Oct. 2017.

\bibitem{echo_state_property}
C.~Gallicchio and A.~Micheli, ``Echo state property of deep reservoir computing
  networks,'' \emph{Cognitive Computation}, vol.~9, pp. 337--350, May 2017.

\bibitem{leaky_integrator}
H.~Jaeger, M.~Lukosevicius, D.~Popovici, and U.~Siewert, ``Optimization and
  applications of echo state networks with leaky-integrator neurons,''
  \emph{Neural Networks}, vol.~20, no.~3, pp. 335--352, 2007.

\bibitem{RL_ESN}
I.~Szita and A.~L. V.~Gyenes, \emph{Reinforcement Learning with Echo State
  Networks}.\hskip 1em plus 0.5em minus 0.4em\relax Springer, Berlin,
  Heidelberg, 2006, vol. 4131.

\bibitem{sutton}
R.~Sutton and A.~Barto, \emph{Introduction to Reinforcement Learning}, 1998.

\bibitem{rician_fading}
A.~Ghaffarkhah and Y.~Mostofi, ``Path planning for networked robotic
  surveillance,'' \emph{IEEE Transactions on Signal Processing}, vol.~60,
  no.~7, pp. 3560--3575, July 2012.

\end{thebibliography}
\vspace{-0.4cm}
\end{document}